\documentclass[12pt]{article}
\usepackage{graphicx}
\usepackage[font=footnotesize]{caption}
\usepackage{natbib}
\usepackage{fancyhdr,amsmath,amssymb, amsthm,mathtools}
\usepackage{etoolbox}

\usepackage[colorlinks,citecolor=blue,urlcolor=blue]{hyperref}
\usepackage{booktabs}
\usepackage{tikz}
\usetikzlibrary{fit,positioning,calc}
\usepackage{scrextend}
\usepackage{hyperref}
\usepackage{chngcntr}
\usepackage{apptools}
\usepackage{bbm}
\usepackage{multirow}
\usepackage[T1]{fontenc}
\usepackage[export]{adjustbox}
\usepackage{mathrsfs}
\usepackage{mathrsfs,dsfont}
\usepackage{amsfonts}
\usepackage{stmaryrd}
\usepackage{mathrsfs,dsfont}
\usepackage{bbm,bm}
\usepackage{algorithm}
\usepackage{algpseudocode}
\usepackage{setspace}
\usepackage{subfigure}
\usepackage{upgreek}
\RequirePackage{xr-hyper}

\allowdisplaybreaks

\AtAppendix{\counterwithin{Lemma}{section}}

\newcommand{\Prob}{\mathbb{P}}

\newcommand{\R}{\mathds{R}}

\newcommand{\E}{\mathds{E}}

\newcommand{\bone}{ {\bf 1} }

\newcommand{\bzero}{ {\bf 0} }

\newcommand{\bZ}{ {\bf Z} }

\newcommand{\bR}{ {\bf R} }
\newcommand{\bY}{ {\bf Y} }

\newcommand{\bD}{ {\bf D} }
\newcommand{\bI}{ {\bf I} }
\newcommand{\bJ}{ {\bf J} }

\newcommand{\bG}{ {\bf G} }

\newcommand{\bC}{ {\bf C} }

\newcommand{\bv}{ {\bf v} }

\newcommand{\bmu}{ {\boldsymbol \mu} }
\newcommand{\bSigma}{ {\boldsymbol \Sigma} }

\newcommand{\bz}{ {\bf z} }

\newcommand{\appropto}{\mathrel{\vcenter{
  \offinterlineskip\halign{\hfil$##$\cr
    \propto\cr\noalign{\kern2pt}\sim\cr\noalign{\kern-2pt}}}}}

\definecolor{green}{rgb}{0.3, 0.73, 0.09}

\newcommand{\magenta}[1]{\!}

\newcommand{\ind}{\,\perp\!\!\!\!\!\perp\,} 

\newtheorem{Lemma}{\bf Lemma}
\newtheorem{Theorem}[Lemma]{\bf Theorem}
\newtheorem{Proposition}[Lemma]{\bf Proposition}

\newtheorem{Remark}[Lemma]{\bf Remark}

\newtheorem{lems}{\textsc{Lemma}}[section]

\makeatletter
\providecommand{\leftsquigarrow}{%
  \mathrel{\mathpalette\reflect@squig\relax}%
}
\newcommand{\reflect@squig}[2]{%
  \reflectbox{$\m@th#1\rightsquigarrow$}%
}
\makeatother

\makeatletter
\newcommand\subsubsubsection{\@startsection{paragraph}{4}{\z@}{-2.5ex\@plus -1ex \@minus -.25ex}{1.25ex \@plus .25ex}{\normalfont\normalsize\bfseries}}
\makeatother

\newcommand{\blind}{1}

\expandafter\def\expandafter\normalsize\expandafter{%
    \normalsize%
    \setlength\abovedisplayskip{4pt}%
    \setlength\belowdisplayskip{4pt}%
    \setlength\abovedisplayshortskip{4pt}%
    \setlength\belowdisplayshortskip{4pt}%
}

\begin{document}

\def\spacingset#1{\renewcommand{\baselinestretch}%
{#1}\small\normalsize} \spacingset{1}

\if1\blind
{
	\title{\vspace{-10pt}
	\Large{\bf Phylogenetic Latent Space Models for Network Data}}
\author{Federico Pavone$^{1}$, Daniele Durante$^{2}$ and Robin J. Ryder$^{3}$\\
$^{1}$CEREMADE, CNRS, Université Paris--Dauphine, Paris\footnote{ The current affiliation of Federico Pavone is Theremia Health, Paris, France. However, this research was conducted while~the~author was a Postdoctoral Fellow at CEREMADE, Université Paris--Dauphine, Paris, France.}\hspace{.2cm}\\
	 $^{2}$Department of Decision Sciences, Bocconi University, Milan \hspace{.2cm}\\
	 $^{3}$Department of Mathematics, Imperial College London, London}
	\date{}
	\maketitle
} \fi

\if0\blind
{
	\bigskip
	\bigskip
	\bigskip
	\begin{center}
		{\LARGE\bf Phylogenetic Latent Space Models for Network Data}
	\end{center}
	\medskip
} \fi

\begin{abstract}
Latent space models for network data characterize each node through a vector of latent features whose pairwise similarities define the edge probabilities among the pairs of nodes. Although this formulation has led to successful implementations, the overarching focus  has been on directly inferring node embeddings through the latent features, rather than learning the generative process underlying these embeddings.  This focus prevents borrowing information across the  node features and limits the ability to infer higher-level architectures governing network formation. For example, routinely-studied networks often exhibit multiscale structures informing on nested modular hierarchies among nodes, which could be learned via tree-based representations of dependencies among the latent features. We pursue this direction by bridging latent variable representations~of network data with concepts from phylogenetic inference to design a novel  latent space model~that  explicitly characterizes the generative process of the node feature vectors through a  branching Brownian motion, with branching structure parametrized by a tree. This tree constitutes the~main object of interest and is learned under a Bayesian perspective leveraging priors inherited from phylogenetic literature to infer tree-based modular hierarchies across nodes, which explain heterogeneous multiscale patterns in the network.  Identifiability~results are derived along with posterior consistency theory. The  inference potentials of our model are illustrated in simulations and two real-data applications from criminology and neuroscience, where our formulation learns core structures hidden to state-of-the-art alternatives.
\end{abstract}
\noindent%

\spacingset{1.78} 

\section{Introduction}\label{sec_1}
Latent variable models for network data characterize the probabilistic process of edge formation~through a function of node-specific latent quantities capable of accounting for core network properties such as, for example, transitivity, stochastic equivalence, homophily and community structure. This broad~family  includes, among others, stochastic block models \citep{nowicki2001estimation}, mixed membership stochastic block models \citep{airoldi2008mixed}, latent space models \citep{hoff2002latent}, random dot-product graph models \citep{athreya2018statistical} and graphon models \citep{caron2017sparse, borgs2018sparse}, thereby providing one of the most widely implemented and studied~classes~of~statistical models for network data. Within this class, latent space models have been the object of primary attention~owing to the associated flexible, yet interpretable, representation, which naturally lends~itself to several extensions in various directions. Focusing on the ubiquitous binary undirected network setting, these models assume the edges among pairs of nodes $(v,u)$ as conditionally independent Bernoulli variables with probabilities depending on a measure of pairwise similarity among vectors of latent endogenous  features characterizing nodes $v$ and $u$, respectively,~and,~possibly, on additional exogenous effects arising from node-specific attributes.~Recalling~\citet{hoff2002latent},~this~representation crucially allows to incorporate central network properties such as transitivity and more nuanced notions of stochastic equivalence that possibly inform on community and homophily structures, thereby~allowing~to~learn~core endogenous architectures through an informative embedding of nodes. This embedding also facilitates the graphical identification of central nodes and improves inference on exogenous node-attribute effects after accounting~for~the~endogenous~network~structure.

These advantages have motivated rapid and successful generalizations of latent space models for a single network in several important directions, including, in particular, dynamic regimes \citep[see, e.g.,][]{sarkar2006dynamic,durante2014nonparametric,sewell2015latent}, multilayer settings \citep[e.g.,][]{gollini2016joint,salter2017latent,macdonald2022latent} and replicated networks contexts \citep[e.g.,][]{durante2017nonparametric,wang2019common,arroyo2021inference}. Although these contributions provide routinely-implemented state-of-the-art extensions of the original formulation proposed by \citet{hoff2002latent}, the overarching focus remains on directly inferring~the nodes latent features, rather than explicitly incorporating  and learning higher-level informative architectures that regulate the formation process of these features. Advancements~in this direction~would~not only improve borrowing of information across the feature vectors~of~the different nodes, but~would~also~open~the avenues for inferring more nuanced structures which characterize complex multiscale network topologies. As will be shown in  Sections~\ref{sec_phy_51} and \ref{sec_phy_52}, these multiscale patterns are inherent~to~several~networks studied in practice, ranging from criminology \citep[see, e.g.,][]{calderoni2017communities,coutinho2020multilevel} to neuroscience \citep[e.g.,][]{bullmore2009complex,bullmore2012economy,betzel2017multi}, and could~be leveraged to learn unexplored tree-based~modular hierarchies across nodes that drive the formation~of such structures. As shown in Figures \ref{figure:crime_2} and \ref{figure:brain_2}, the inferred tree-based structures  unveil, for example, hidden organizational architectures of modern criminal networks and fundamental multiscale modules~in~structural~brain~connectivity.

\begin{figure}[t!]
	\centering
	\includegraphics[width=1\textwidth]{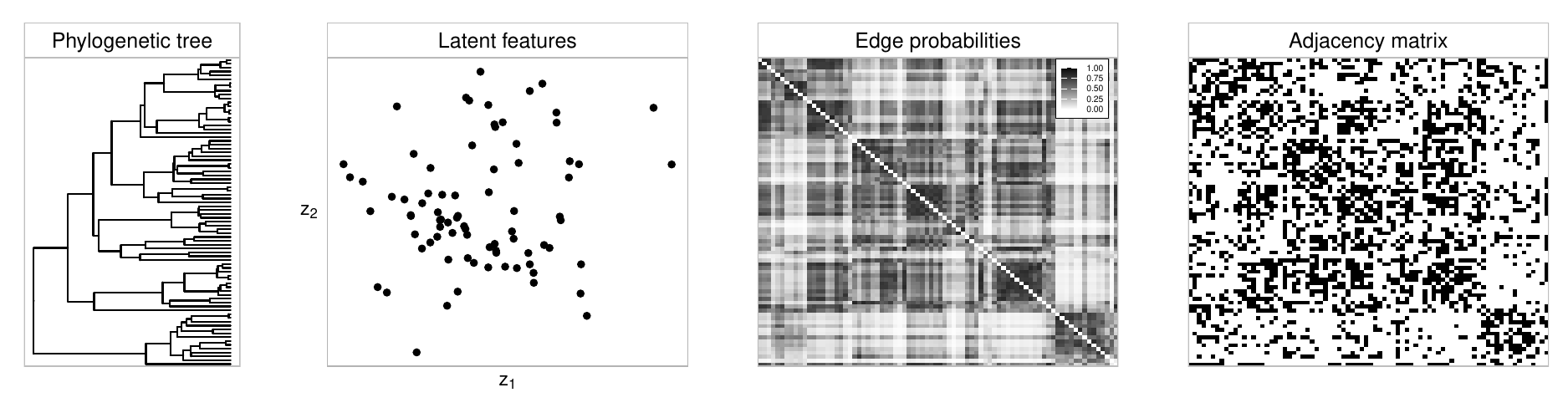}
		\vspace{-15pt}
    \caption{{Graphical representation of the generative process underlying the proposed phylogenetic latent space model. From left to right: Tree defining the branching structure regulating the feature formation process, modeled~via~a~branching Brownian motion (the leaves of the tree correspond to nodes of the network); Node-specific latent features obtained as a realization of the  branching Brownian motion; Matrix of pairwise edge probabilities defined as the logit mapping of the negative Euclidean distance among~the pairs of~node-specific latent features; Adjacency matrix representation of the network with edges sampled from independent Bernoullis conditioned on the corresponding edge probability. Although  \textsc{phylnet} allows for $K$-dimensional  feature vectors, here we consider $K=2$ to facilitate visualization. }}
    \label{figure:phylnet}
\end{figure} 

Our aim is to address the above gap and infer unexplored tree-based hierarchies between nodes.~To this end, we develop a novel phylogenetic latent space~model (\textsc{phylnet}) for network data  that explicitly characterizes the mechanism of the node~features formation via a branching Brownian motion, with branching structure parametrized by a tree whose prior is inherited from phylogenetic inference \citep[e.g.,][]{felsenstein2004inferring}. 
As illustrated in Figure~\ref{figure:phylnet}, the generative process behind the proposed model combines classical latent~space representations of the edges in the network \citep{hoff2002latent} with a structured characterization of~the  feature formation process through a flexible tree-based architecture, which can unveil increasingly-nested modular hierarchies among the nodes. 
Although this perspective has been mostly overlooked in the literature on latent space models, we note in Remark \ref{rem1} that classical formulations \citep{hoff2002latent} relying~on independent and identically distributed Gaussian priors for the latent features can be, in fact, interpreted as a degenerate version of our model, where the tree structure is not learned, but rather implicitly fixed at a restrictive topology. This coincides~with~an~ultrametric tree where all the $V$ nodes split at the root, so as to generate $V$ separate  Brownian motions yielding independent Gaussian priors at the leaves with equal variances guaranteed by the ultrametric property (i.e., all the paths connecting the root to a leaf have equal length). While convenient,~this~representation relies on~a restrictive and pre-specified topology, which is unable to characterize and learn more nuanced hierarchical architectures among nodes and increasingly-nested group structures that ultimately drive the formation of multiscale network patterns observed in practice.~Addressing~this~issue requires a formulation that (i) does not pre-specify a tree structure,~but~rather~treats~it~as~an~object to be inferred, and (ii) can flexibly account for complex topologies, which naturally enlarge the simplified one which is implicit in classical latent space models \citep{hoff2002latent}. The~proposed \textsc{phylnet} construction addresses (i) by explicitly including the tree as a model parameter to be inferred, and achieves (ii) through priors that are routinely employed in the filed of phylogenetic inference \citep{felsenstein2004inferring} to obtain full support in~the~entire~space~of~binary~ultrametric~trees.

Our perspective is also supported  by past and recent research on non-Euclidean geometries~of~the latent space that can accurately characterize more nuanced network structures often observed in~practice \citep[see, e.g.,][]{smith2019geometry, lubold2023identifying}; see also \citet{freeman1992sociological, borgatti1990ls} and \citet{schweinberger2003settings}  for a specific focus on ultrametrics capable of accounting for hierarchies~of nested transitive relations. These contributions highlight an inherent connection between~several non-Euclidean geometries and tree-based structures.~However, despite the importance of this connection, there have been limited efforts to develop~a flexible, yet~interpretable,  latent space model characterizing explicitly the feature formation mechanism via~a tree-based representation that can be learned as part of the inference process. In fact,  current attempts to include and learn structure in the features' formation process \citep[e.g.,][]{handcock2007model, fosdick2019multiresolution} mostly focus on combining ideas from latent~space models and stochastic block models to infer not only lower-level heterogeneous node-specific~latent~features, but also higher-level group structures related to such features. While this perspective provides an important improvement over classical latent space models,~as clarified in  Sections~\ref{sec_phy_51} and \ref{sec_phy_52}, this two-level representation is not designed to infer more~complex~hierarchical architectures among nodes leading to those multiscale network patterns often observed in practice. The importance of accounting for these more nuanced mechanisms is evident in the recent contribution by \citet{pu2025tree}, where the formation of node feature  vectors is supervised by external information~available in the form of a tree-structured node hierarchy. However, besides relying on a substantially different generative mechanism relative to the one we design for the proposed \textsc{phylnet} model, this construction assumes the tree as a known external object, rather than as an unknown architecture~to~be learned. Therefore, there is no guarantee that the employed exogenous~tree-structured node hierarchy aligns with the endogenous ones driving network formation.  Unveiling these~latter~structures requires treating the tree as a model parameter to be inferred within a wide class of topologies. The \textsc{phylnet} formulation is specifically designed to address such an objective.  

Although the direction we pursue  is novel in the latent space modeling framework,  inference on~tree-based representations has been explored in generalizations of stochastic block models  \citep[e.g.,][]{roy2006learning,clauset2008hierarchical,roy2008mondrian,herlau2012detecting}. However, the overarching focus  has been on learning endogenous hierarchies between nodes or node partitions, under stochastic equivalence or more general homogeneity assumptions. This perspective yields substantially different models and tree-based representations. Moreover, it might fail to characterize heterogeneous patterns at the node level and, as illustrated empirically in Section~\ref{sec_phy_4}, it experiences~challenges~in networks without~a~clear block structure. Conversely, the  \textsc{phylnet} model flexibly characterizes broader network structures by modeling the feature formation process directly at the node level. This requires methodological and computational innovations also with respect to the literature on phylogenetic trees, since~the~node-specific features are not directly observed, but~rather denote latent quantities that are identifiable~only up~to translations and rotations  \citep{hoff2002latent}. In Section~\ref{indet}, we prove theoretically that, despite this identifiability issue at the latent features level, the main object~of inference, i.e., the tree, remains identifiable. 

The \textsc{phylnet} model also lends itself naturally~to extensions~in different settings, including~latent space representations for directed, bipartite, multilayer and multiplex/replicated networks. Section~\ref{sec_phy_23} pursues this direction with a focus on the latter class, which comprises~multiple~network observations, on the same set of nodes, encoding either different types of relationships~(multiplex networks) or replicated measurements of the same connectivity  notion  (replicated networks). In this context, we allow the latent features~of~each~node~to change across networks, but assume a shared tree  regulating the formation of these features~in~the different networks. This allows to borrow information between the networks for estimating complex tree-based modular hierarchies among nodes, while preserving flexibility in characterizing network-specific multiscale structures. In this setting we also prove theoretical consistency properties in learning the  tree as the number of networks increases;~see~Section~\ref{sec_consis}.

In Section~\ref{sec_phy_3}, we provide details for posterior inference on the tree parameterizing~the~\textsc{phylnet} model in a Bayesian setting and under a pure-birth process prior for the tree inherited from phylogenetic inference. Combined with the latent space model likelihood, such a prior induces a posterior distribution for the tree, which we learn through a bespoke Metropolis-within-Gibbs~algorithm.  The tree samples produced by this algorithm are then  summarized~for~both~point~estimation~and uncertainty quantification via the notions of {\em consensus tree} \citep[][]{felsenstein2004inferring} and {\em DensiTree} \citep{bouckaert2010densitree}. 
Our simulation studies in Section~\ref{sec_phy_4} show that \textsc{phylnet} provides remarkable performance improvements over both heuristic and model-based solutions for learning hierarchical multiscale structures~in network data, including state-of-the-art generalizations of stochastic block models  \citep[e.g.,][]{clauset2008hierarchical}. Accounting for such structures further facilitates borrowing of information among node latent features to obtain improved estimates of the edge probabilities under \textsc{phylnet}, compared to classical latent space models \citep[e.g.,][]{hoff2002latent}. These gains are confirmed in applications to multiple~corrupted measurements of a Mafia-type criminal network \citep[e.g.,][]{calderoni2017communities,legramanti2022extended} (see Section~\ref{sec_phy_51}), and to replicated brain connectivity data observed from different individuals \citep[e.g.,][]{craddock2013imaging,kiar2017high} (see Section~\ref{sec_phy_52}). In the former case, \textsc{phylnet}  unveils a previously-unexplored tree-based reconstruction of the organizational architecture of the Mafia group under analysis, and highlights criminals with highly peculiar positions within the hierarchy.~In~the~latter, the inferred  tree learns interesting~nested~symmetries in the human brain, while  pointing toward a top-level frontal-back division of the brain followed by a nested partition in the two hemispheres.~Future directions of research are discussed in Section~\ref{sec_phy_6}. Proofs and additional results can be found~in~the Supplementary~Material.

\section{Phylogenetic Latent Space Models}\label{sec_phy_2}

\subsection{A brief introduction to phylogenetic trees}\label{sec_phy_21}
As discussed in Section~\ref{sec_1}, the proposed  \textsc{phylnet} model describes the formation of node features~through a tree-structured mechanism, whose prior is inherited from phylogenetic inference. More specifically, the trees we consider are rooted, binary branching trees endowed with branch~lengths. As mentioned in our  introduction, this choice naturally extends the tree structure that is implicit in classical latent space models \citep{hoff2002latent} to explore substantially more general architectures gaining flexibility in the positioning of the internal splits and in the length of the~tree branches. When endowed~with a prior, these architectures can be characterized as random trees with random branch lengths~or~equivalently as branching processes \citep[][]{aldous2001stochastic}.~By~providing~an interpretable graphical representation of tree-structured relationships among entities, these constructions have been  extensively studied and routinely employed to infer evolutionary relations in phylogenetic inference; see e.g., \citet{felsenstein2004inferring} for an introduction. Nonetheless, more recent contributions have shown that the utility of these ideas extends  beyond evolutionary biology, yielding successful implementations also in  linguistics \citep{hoffmann2021bayesian,phylolinguisticsworkflow} and cultural evolution \citep{evans2021uses, buckley2025contrasting}.

In Sections~\ref{sec_phy_22}--\ref{sec_phy_23}, we further extend the applicability and potentials of this construction~by~leveraging trees to characterize the formation process of node-specific latent features in models for network data. To this end, we consider trees having a fixed number of leaves $V$ (corresponding to the nodes within the network) and meeting the ultrametric property (all paths connecting~the~root~to~a~leaf~node have the same length, defined as the sum of the branch lengths). As discussed~in Section~\ref{sec_1},~this~property is  motivated~by the relevance of specific non-Euclidean geometries in latent space models for networks \citep[e.g.,][]{smith2019geometry}, and is met by the tree structure that is implicit in the latent space model as originally introduced by \citet{hoff2002latent} (refer to Section~\ref{sec_phy_6} for possible extensions to non-ultrametric settings). To perform inference on such~a~tree under the  \textsc{phylnet} model, we employ a Bayesian~perspective and leverage priors that are widely used in the phylogenetics literature \citep[see, e.g.,][]{chen2014bayesian}. We consider in particular elicitations within the class of   pure~birth Yule processes \citep{yule1925ii}. As clarified in the following, this simple and interpretable construction not only facilitates principled point estimation and uncertainty quantification, but also achieves full support on the space~of~trees that are of interest and can be reasonably learned in latent space models for networks. Note~that we restrict ourselves to binary trees, a common choice in phylogenetics. This is not a strong~constraint,~as a multifurcating tree can be approximated by a binary tree with very short branches. Moreover, the consensus trees we use as a summary of the posterior distribution need~not~be~binary

In classical phylogenetic models, the above tree constitutes the branching architecture~that~regulates the formation of a certain character of interest across a set of entities denoting the leaves~of~the tree. 
Often, such a process is modeled as either a continuous or discrete state-space continuous-time Markov chain, depending on the nature of the character.  In our \textsc{phylnet} model, the character denotes the real-valued latent features of the different nodes, thereby requiring a  continuous~state-space process. Motivated by its success in evolutionary modeling \citep[][]{felsenstein1985phylogenies,eastman2011novel,may2020bayesian} we consider, in particular, a branching Brownian motion to provide a flexible, yet tractable, characterization of the feature formation process over the tree architecture. This choice crucially induces a Gaussian distribution for the features at the leaves, with a variance-covariance structure that reflects~the tree topology.  For~simplicity, consider the case in which each leaf $v=1, \ldots, V$ is characterized by a single latent feature $z_{v} \in \mathbb{R}$, and  denote with $\sigma^2$ the variance parameter of the Brownian~motion. Let $T$ be the height of the tree $\Upsilon$, \textit{i.e.,} the distance between the root of the tree and any of its leaves. For every pair of nodes $(v,u)$, let $t_{vu}\in[0,T]$ be the distance between~the~root and the most recent common ancestor to $v$ and $u$; a large value of $t_{vu}$ means that $v$ and $u$ are close in the tree. Then, under the branching Brownian motion model, the joint distribution of the features' vector  at the leaves of the tree is a multivariate Gaussian,  \textit{i.e.,}  $[z_{1}, z_{2}, \ldots, z_{V}]\sim\mbox{N}_V(0,\Sigma_{\Upsilon})$ with marginal variances $\mbox{var}(z_{v})=\sigma^2 T$ and covariances $\mbox{cov}(z_{v},z_{u})=\sigma^2 t_{vu}$.  Hence, the ``closer'' $v$ and $u$ are in the tree, the higher the covariance among the corresponding features. This means that tree-based representations~of nested modular hierarchies among nodes can be possibly learned from the structure in the pairwise dependencies among nodes latent features. This intuition~is~at~the~basis~of~the~\textsc{phylnet}~model~presented in Section~\ref{sec_phy_22} (see Chapter 4 in \citet{felsenstein2004inferring} for more details on phylogenetic trees).

\setlength\abovedisplayskip{8pt}%
    \setlength\belowdisplayskip{8pt}%

\subsection{Model formulation for a single network}\label{sec_phy_22}
\vspace{4pt}
Although the   \textsc{phylnet} model can be readily extended to several network structures, including directed, bipartite and weighted settings, we focus here on the simplest, yet ubiquitous,~case~of a single binary undirected network. Let $\bY$ denote the $V\times V$ symmetric adjacency matrix representation~of~such a network, so that  $[\bY]_{vu}=y_{vu}=y_{uv}=1$ if nodes $v$ and $u$ are connected, and $y_{vu}=y_{uv}=0$ otherwise, for $1\leq u < v \leq V$. Consistent with classical latent space models for network data \citep{hoff2002latent}, we assume  
\begin{equation}\label{eq:lik_single}
(y_{vu}\mid \theta_{vu}) \overset{\text{ind}}{\sim}\textsc{bern}(\theta_{vu}) \quad \mbox{with} \quad    \text{logit}(\theta_{vu}) = a\, - \lVert\bz_v - \bz_u \lVert, \qquad  1\leq u < v \leq V,
\end{equation}
where $a \in \mathbb{R}$ is a scalar controlling the overall network density, while $\bz_v \in \mathbb{R}^{K}$ and $\bz_u \in \mathbb{R}^{K}$~are the vectors of latent features for the nodes $v$ and $u$, respectively, with $\lVert\bz_v - \bz_u \lVert$ denoting~the~Euclidean~distance among the features vectors. The above representation embeds the nodes into a $K$-dimensional latent space with positions informing on the edge probabilities. More specifically, the closer the features vectors of  $v$ and $u$, the higher the probability to observe an edge among these two nodes.~As~mentioned in Section~\ref{sec_1}, such a natural interpretation has motivated a direct focus on the latent features $\bz_1, \bz_2, \ldots, \bz_V$ as the main object of inference. To this end, routine implementations rely on  independent Gaussian \citep{hoff2002latent} or mixtures of Gaussian \citep{handcock2007model} priors for the node latent features and then provide inference on the induced posterior distribution. Although the latter perspective increases flexibility and inference potentials relative to the former, both solutions lack a structured and realistic characterization of dependence among the node-specific feature vectors, and still treat these vectors as the main object of interest. As a consequence, inference reduces to graphical interpretations of a model-based embedding of the nodes, possibly grouped in dense communities under mixtures of Gaussian priors. In fact, one would expect that the latent features exhibit, in practice, structured~dependence~across~nodes, with this dependence possibly unveiling more fundamental and informative architectures, such as increasingly nested modular hierarchies among nodes,  at the basis of multiscale patterns in the network.

Motivated by the above discussion and recalling Section~\ref{sec_phy_21}, we complement model~\eqref{eq:lik_single}~with a structured characterization of the feature formation process via a branching Brownian motion parameterized by a tree that constitutes the main object of inference and is assigned a prior inherited from phylogenetic literature. Crucially, this advancement (i) facilitates an improved reconstruction of node embeddings by explicitly borrowing~information across the features of the different nodes and (ii) substantially enlarges inference potentials on hidden node hierarchies through the  tree. Specifically, let $\bZ \in\R^{K\times V}$ be the $K \times V$ matrix with generic row $\bZ_{[k]}=[z_{k1}, \ldots, z_{kV}] \in \R^{1 \times V}$ encoding the values of the $k$-th feature for the~$V$~nodes, and denote with $\Upsilon$ the (ultrametric) tree characterizing the branching architecture which regulates features' formation.  In this setting and conditional on $\Upsilon$, we assume $K$ independent branching Brownian motions (\textsc{bbm}) for the different features' dimensions from $1$ to $K$. Recalling Section~\ref{sec_phy_21}, this implies that, at the leaf nodes, 
\begin{equation}\label{eq:bbm_row}
    \begin{aligned}
    (\bZ^{\intercal}_{[k]}\mid \mu_k, \sigma^2,\Upsilon) \overset{\text{ind}}{\sim}\mbox{N}_V(\mu_k\bone_V,\sigma^2\bSigma_\Upsilon), \qquad \text{for}\,\, k=1,\dots,K,
    \end{aligned}
\end{equation}
where $\bone_V$ is the $V\times 1$ vector of all ones, $\mu_k$ denotes a centering parameter for the $k$-th feature,~$\sigma^2$ is the \textsc{bbm} rate, while $\bSigma_\Upsilon$ corresponds to the $V\times V$ correlation matrix induced by the tree $\Upsilon$,~as described in Section~\ref{sec_phy_21}.
In the applications we consider, the parameter $T$ (i.e., the height of $\Upsilon$)~has no clear interpretation and is not identifiable. As such, we  fix $T=1$, so that $\sigma^2>0$ can be directly~interpreted as the marginal variance of each $z_{kv}$, for $v=1, \ldots, V$, while $[\bSigma_\Upsilon]_{vu}=t_{vu} \in [0,1]$ denotes the correlation between $z_{kv}$ and $z_{ku}$ for any $k=1,\dots,K$. Let us emphasize that setting $T=1$ does not constrain the prior specification in~\eqref{eq:bbm_row}, nor does it affect the tree topology.

\setlength\abovedisplayskip{6.3pt}%
    \setlength\belowdisplayskip{6.3pt}%
    
To complete the Bayesian specification, we require priors for  $\sigma^2$, $\mu_1, \ldots, \mu_K$ and the tree~$\Upsilon$~in~\eqref{eq:bbm_row}, along with the scalar $a$ in \eqref{eq:lik_single}. Regarding $\sigma^2$ and $\mu_1, \ldots, \mu_K$, note that, under \eqref{eq:bbm_row}, the latent features are identifiable only up to translations and rotations  \citep[see, e.g.,][]{hoff2002latent}. As proved in~the~following, this issue does not affect the identifiability of the main object of interest, i.e., the tree  $\Upsilon$, but  it implies that the centering quantities $\mu_1, \ldots, \mu_K$ in \eqref{eq:bbm_row} can be treated as nuisance parameters. However,~as highlighted in Section~\ref{sec_phy_31}, including these quantities allows us to re-center the latent features at~each~step of the Metropolis-within-Gibbs routine we derive, and hence obtain improved mixing and convergence of the chain for $\sigma^2$. Consistent with this discussion, we consider diffuse Gaussian priors for~$\mu_1, \ldots, \mu_K$ and a conditionally conjugate inverse-Gamma prior for $\sigma^2$. More precisely, we let
\begin{equation}\label{eq:bbm_mu_sig}
    \begin{aligned}
   \mu_k \overset{\text{ind}}{\sim} \mbox{N}(0, \sigma_{\mu}^2), \ \ \text{for}\,\, k=1,\dots,K, \qquad \qquad \sigma^2 \sim \textsc{inv-gamma}(\alpha_\sigma,\beta_\sigma).
    \end{aligned}
\end{equation}
Recalling Section~\ref{sec_phy_21}, for the tree  $\Upsilon$ we consider a  pure-birth Yule process  prior \citep{yule1925ii}.~This~can be derived as a special case of birth and death branching processes $\textsc{bdt}(b,d)$ \citep[e.g.,][]{harris1963theory,ross2014introduction} employed in phylogenetics by setting the death~rate~$d=0$ to obtain
\begin{equation}\label{eq:priortree}
    \begin{aligned}
   (\Upsilon \mid b) \sim \textsc{bdt}(b,0), \qquad b \sim \textsc{inv-gamma}(\alpha_b,\beta_b).
    \end{aligned}
\end{equation}
This is a common choice  in phylogenetic inference, having the advantage of being both tractable and fully supported on the space of binary ultrametric trees. This allows flexible Bayesian learning~of~different, possibly complex, tree architectures, while facilitating posterior computation.

Finally, for the scalar $a$ in \eqref{eq:lik_single}, we follow standard practice  and let
\begin{equation}\label{eq:priora}
    \begin{aligned}
   a \sim \mbox{N}(0,\sigma^2_a),
    \end{aligned}
\end{equation}
as in, e.g., \citet[][]{hoff2002latent}.

\vspace{-1pt}

\begin{Remark}\label{rem1}
\itshape As mentioned in Section~\ref{sec_1}, the above formulation provides a natural extension~of~the~latent space model as originally proposed in \citet[][]{hoff2002latent}. In fact, the original latent space model can~be obtained as a special case of our representation by replacing \eqref{eq:priortree} with a deterministic pre-specified tree $\Upsilon$ where all $V$ nodes split at the root so that $t_{vu} =0$~for~any pair $(v,u)$, and hence,~$\bSigma_\Upsilon$~in~\eqref{eq:bbm_row}~reduces to the identity (recall that \eqref{eq:priortree} can produce trees arbitrarily close to this constraint).  Thus the  \textsc{phylnet}  construction enlarges routinely-employed latent space formulations not only from a modeling perspective, but also in terms of inference potentials on more nuanced tree-based generative structures, well-beyond node latent features.
\end{Remark}

\begin{Remark}\label{rem2}
\itshape Although the  \textsc{phylnet}~model represents a genuine transfer of methods from phylogenetics to latent variable representations of networks, the tree $\Upsilon$ shall not be interpreted~under the evolutionary perspective of phylogenetics. Rather, within our formulation, it encodes~structured dependence among nodes in a network through a flexible architecture that characterizes nested modular hierarchies between nodes, along with the induced multiscale connectivity patterns.
\end{Remark}
\vspace{-1pt}

Equations \eqref{eq:lik_single}--\eqref{eq:priora} formalize the proposed  \textsc{phylnet} representation and provide a flexible generative~latent space characterization of complex multiscale network structures that unveil informative tree-based modular hierarchies among the nodes, encoded in $\Upsilon$. Learning this parameter under  \eqref{eq:lik_single}--\eqref{eq:priora} is, in principle, possible by  combining results from classical latent space models for network data  \citep{hoff2002latent} and Bayesian phylogenetic inference \citep[see, e.g.,][]{chen2014bayesian}. However,~as~previously~discussed, unlike for classical evolutionary models, in the   \textsc{phylnet} formulation~the~tree regulates the formation of features that are not directly observed, but rather latent and identifiable only  up to  translations~and rotations \citep{hoff2002latent}. This motivates innovations in terms of posterior~computation~(see~Section~\ref{sec_phy_31}) and, more crucially, requires~theory guaranteeing that inference~on~the~tree $\Upsilon$ parameterizing the formation process of the latent features is not affected by the identifiability~issues of these features in the likelihood.

\vspace{5pt}

\subsection{Identifiability of the tree $\Upsilon$}\label{indet}
Consistent with the discussion above, Theorem~\ref{thm:s2_g} and Proposition~\ref{prop:matrix_normal} clarify that, despite the identifiability issues of the latent features, the tree  $\Upsilon$ and the rate $\sigma^2$ of the Brownian motion~remain identifiable. These two results rely on Lemmas \ref{prop:constant} and \ref{thm:ident_lat} below, which state properties of direct interest for the broader class of latent space models  \citep{hoff2002latent}, beyond \textsc{phylnet}. More specifically, Lemma~\ref{thm:ident_lat} exploits a geometrical result for Euclidean distance matrices  \citep[e.g.,][]{hayden1991cone} provided in Lemma \ref{prop:constant}, to show that  the scalar $a$ and the $V\times V$ matrix $\bD$ of distances $[\bD]_{vu}=d_{vu}= \lVert\bz_v - \bz_u \lVert$ are identifiable under \eqref{eq:lik_single}. See~the~Supplementary Material for proofs.

\begin{Lemma}\label{prop:constant}
\itshape Consider a $V\times V$ Euclidean distance matrix $\bD^2$ which admits a representation in $\R^K$,~i.e., there exist $\bz_1,\dots,\bz_V\in \R^K$ such that $\forall v,u, [\bD^2]_{vu} =  \lVert\bz_v - \bz_u \lVert^2$. Let $\delta$ be the minimum non-zero value in $\bD$, i.e. $\delta = \min_{u\neq v} [\bD]_{vu}$.
Let $\widetilde{\bD}$ be the $V\times V$ matrix defined as
\begin{equation} \label{D_id}
    \begin{aligned}
        [\widetilde{\bD}]_{vu} = [\bD]_{vu} + c \ \ \ \text{if}\  v\neq u, \qquad \quad
        [\widetilde{\bD}]_{vv} = 0,
    \end{aligned}
\end{equation} 
for some $c > - \delta$ and $c\neq0$. Then, if $V>2K+1$, $\widetilde{\bD}^2$ does not admit a representation in $\R^{K}$.
\end{Lemma}
\vspace{5pt}

\begin{Lemma}\label{thm:ident_lat}
\itshape Consider the latent space model within \eqref{eq:lik_single} for the $V\times V$ adjacency matrix $\bY$, and~let $\bD$ denote~the $V\times V$  pairwise distance matrix having entries $[\bD]_{vu}=\lVert\bz_v - \bz_u \lVert$, for every  $v,u=1, \ldots, V$. Then, if $V > 2K + 1$,~the~parameters $a$ and $\bD$ are identifiable. More precisely, denoting with $\Prob_{a,\bD}$ the joint model for the edges in the adjacency matrix $\bY$, it holds that,
\begin{equation}
    \Prob_{a,\bD} \overset{\text{d}}{=} \Prob_{\widetilde{a},\widetilde{\bD}} \implies (a,\bD) = (\widetilde{a},\widetilde{\bD}),
\end{equation}

    \vspace{-5pt}
\noindent where $\overset{\text{d}}{=}$ denotes equality in distribution.
\end{Lemma}

Lemma~\ref{thm:ident_lat} states a relevant identifiability result for generic latent space models. However, the tree  $\Upsilon$ and the rate $\sigma^2$ of the Brownian motion do not parameterize ${\bD}$ directly, but rather~the~features that define the distances in  ${\bD}$. Theorem~\ref{thm:s2_g} and Proposition~\ref{prop:matrix_normal} combine the result in Lemma~\ref{thm:ident_lat}~with~properties of binary trees and multivariate Gaussians to prove that  $\Upsilon$  and  $\sigma^2$  are also identifiable. 

\begin{Theorem}\label{thm:s2_g}
\itshape	Let  \smash{$\Prob_{\sigma^2,\Upsilon}$} be the  distribution of a $V\times V$ symmetric adjacency matrix $\bY$ whose binary entries are distributed~as in \eqref{eq:lik_single}, conditional on parameters $\sigma^2$ and $\Upsilon$ (with $\bZ$ integrated out under~\eqref{eq:bbm_row}).  Then, the parameters $\sigma^2$ and $\Upsilon$ are identifiable. More specifically,
    \begin{equation}
        \Prob_{\sigma^2,\Upsilon} \overset{\text{d}}{=} \Prob_{\widetilde{\sigma}^2,\widetilde{\Upsilon}} \implies (\sigma^2,\Upsilon) = (\widetilde{\sigma}^2,\widetilde{\Upsilon}),
    \end{equation}
    
    \vspace{-5pt}
\noindent where $\overset{\text{d}}{=}$ denotes equality in distribution. 
\end{Theorem}

Although Theorem~\ref{thm:s2_g} guarantees identifiability of ($\sigma^2, \Upsilon$) under the marginalized model \smash{$\Prob_{\sigma^2,\Upsilon}$},~in practice, the  Metropolis-within-Gibbs routine in Section~\ref{sec_phy_31} samples also the latent features in  $\bZ$ to~facilitate the derivation of tractable full-conditional distributions for $\Upsilon$ and $ \sigma^2$. However,~as discussed previously, the latent features are identifiable only~up~to translations and rotations.~In classical latent space models \citep[see, e.g.,][]{hoff2002latent}, where the focus  is specifically on these latent features, this issue is addressed through a processing  that aligns the samples of  $\bZ$ via  Procrustean transformation. Crucially, such a  re-alignment~is~not~required when the focus of inference is on $\Upsilon$ and $ \sigma^2$, as in the   \textsc{phylnet} model.~In~particular, the translation aspect is addressed through the re-centering operated by the auxiliary parameters~$\mu_1, \ldots, \mu_K$ in \eqref{eq:bbm_row}. As for rotation, Proposition~\ref{prop:matrix_normal}  guarantees that the conditional distribution of the zero-centered~features~given~$\Upsilon$~and~$ \sigma^2$~is~invariant  with respect to orthogonal transformations.

\begin{Proposition}\label{prop:matrix_normal}
\itshape Let $\bZ$ be the $K\times V$ matrix with rows \smash{$\bZ_{[1]},\dots,\bZ_{[K]}\in \R^{1 \times V}$} distributed as in~\eqref{eq:bbm_row}.~Then, for any $K\times K$ orthogonal matrix $\bR$, it holds that 
$$( \bR(\bZ-\bmu \otimes \smash{{\bf 1}^\intercal_{V}} )\mid \sigma^2,\Upsilon) \overset{\text{d}}{=} (\bZ-\bmu \otimes \smash{{\bf 1}^\intercal_{V}} \mid \sigma^2,\Upsilon),$$
 where $\bmu=(\mu_1, \ldots, \mu_K)^{\intercal}$ and \  $\overset{\text{d}}{=}$ \ denotes equality in distribution.
\end{Proposition}

\subsection{Extension to multiple network measurements}\label{sec_phy_23}

Before deriving in Section~\ref{sec_phy_3} the Metropolis-within-Gibbs routine for inference under the~\textsc{phylnet}~model, we generalize in this section the  formulation  to the case of multiple adjacency matrices  $\bY^{(1)}, \ldots, \bY^{(M)}$, each encoding different connections between the same set of nodes. In this context, which includes modern settings of direct interest, such as multiplex and replicated networks  \citep[e.g.,][]{kivela2014multilayer}, the proposed \textsc{phylnet} model not only admits~a natural generalization, but also benefits from the multiple network measurements to achieve an improved reconstruction of the underlying  tree $\Upsilon$, with theoretical guarantees of posterior consistency.

Consistent with other extensions of latent space models to multiple~network~settings \citep[][]{gollini2016joint,durante2017nonparametric,salter2017latent,wang2019common, arroyo2021inference,macdonald2022latent}, we generalize the  \textsc{phylnet} model in a way that achieves flexibility in modeling differences between the multiple networks, while learning relevant shared structures. Since the nodes are common across these networks, within our novel framework it is natural to expect that the architecture (i.e., the tree $\Upsilon$) underlying the feature~formation process is shared among  $\bY^{(1)}, \ldots, \bY^{(M)}$, with the  differences among these $M$ adjacency matrices being the result of network-specific features $\bZ^{(1)}, \ldots, \bZ^{(M)}$ evolving over  the shared~tree $\Upsilon$ . Motivated by these considerations, we let
\begin{equation}\label{eq:lik}
(y^{(m)}_{vu}\mid \theta^{(m)}_{vu}) \overset{\text{ind}}{\sim}\textsc{bern}(\theta^{(m)}_{vu}) \ \ \mbox{with} \ \    \text{logit}(\theta^{(m)}_{vu}) = a\, - \lVert\bz^{(m)}_v - \bz^{(m)}_u \lVert, \quad 1\leq u < v \leq V,
\end{equation}
independently for $m=1, \ldots, M$, where $y^{(m)}_{vu}=[\bY^{(m)}]_{vu}$, with $y^{(m)}_{vu}=y^{(m)}_{uv}$, while  $\bz^{(m)}_v \in \smash{\mathbb{R}^{K}}$~and \smash{$\bz^{(m)}_u \in \mathbb{R}^{K}$} denote~the features of nodes $v$ and $u$, respectively, in the $m$-th network. As discussed~above, these features are allowed to change across the networks to flexibly characterize differences in the edge probabilities between the adjacency matrices $\bY^{(1)}, \ldots, \bY^{(M)}$ through~a formation mechanism regulated by a shared tree architecture $\Upsilon$. Extending \eqref{eq:bbm_row}--\eqref{eq:priora} to this setting, yields
\begin{equation}\label{eq:bbm_row_multi}
    \begin{aligned}
    (\bZ^{(m)\intercal}_{[k]}\mid \mu^{(m)}_k, \sigma^2,\Upsilon) \overset{\text{ind}}{\sim}\mbox{N}_V(\mu^{(m)}_k\bone_V,\sigma^2\bSigma_\Upsilon), \qquad \text{for}\,\, k=1,\dots,K, \ \ m=1, \ldots, M,
    \end{aligned}
\end{equation}
where the locations and rate  have independent~priors
\begin{equation}\label{eq:bbm_mu_sig_multi}
    \begin{aligned}
  \quad  \mu^{(m)}_k \overset{\text{ind}}{\sim} \mbox{N}(0, \sigma_{\mu}^2), \ \ \text{for}\,\, k=1,\dots,K, \ m=1, \ldots, M, \quad \ \sigma^2 \sim \textsc{inv-gamma}(\alpha_\sigma,\beta_\sigma),
    \end{aligned}
\end{equation}
while the tree $\Upsilon$ and the scalar $a$ are assigned the same priors as in Section~\ref{sec_phy_22}, namely
\begin{equation}\label{eq:priortree_multi}
    \begin{aligned}
   (\Upsilon \mid b) \sim \textsc{bdt}(b,0), \quad b \sim \textsc{inv-gamma}(\alpha_b,\beta_b), \ \ \qquad  \quad a \sim \mbox{N}(0,\sigma^2_a). \qquad \qquad \ \ \
    \end{aligned}
\end{equation}
Model \eqref{eq:lik}--\eqref{eq:priortree_multi} is an effective extension of the  \textsc{phylnet} construction  to multiple networks (note that this extension can be made even more flexible, without major computational  complications, by letting both $\sigma^2$ and $a$ be network-specific, i.e., $\sigma^{(1)}, \ldots, \sigma^{(M)}$~and $a^{(1)}, \ldots, a^{(M)}$).

\subsection{Posterior consistency for the tree $\Upsilon$}\label{sec_consis}

Model \eqref{eq:lik}--\eqref{eq:priortree_multi} not only achieves flexibility in modeling multiple networks, but also benefits from these $M$ measurements to gain efficiency in inference~on~the~shared tree $\Upsilon$. This property is formalized in Theorem~\ref{thm:consist}, which states consistency for the posterior  on $(\sigma^2,\Upsilon)$ as $M \rightarrow \infty$.

\begin{Theorem}\label{thm:consist}
\itshape Let the tree $\Upsilon=(\mathcal{T},\bm{\lambda})$ be decomposed into its  tree topology \smash{$\mathcal{T}\in\mathbb{T}_V$}, where $\mathbb{T}_V$ is the space of binary tree topologies with $V$ terminal nodes, and the collection of branch lengths $\bm{\lambda}\in\mathbb{R}_{+}^{2V-2}$ satisfying the ultrametric property. Then, it holds that $\forall \mathcal{T}_0\in\mathbb{T}_V$, there exists $H_0\subset \mathbb{R}_+\times\mathbb{R}_{+}^{2V-2}$ with \smash{$\Pi_{\sigma^2,\bm{\lambda}\mid\mathcal{T}_0}(H_0)=1$}, such that $\forall(\sigma^2_0,\bm{\lambda}_0)\in H_0$ if $\bY_1,\bY_2,\dots\sim \smash{\mathbb{P}_{\sigma_0^2,\Upsilon_0}}$ i.i.d. with $\Upsilon_0=(\mathcal{T}_0,\bm{\lambda}_0)$, then for any neighborhood $B$ of \smash{$(\sigma_0^2,\Upsilon_0)$}, 
\begin{equation}
    \lim_{M\rightarrow \infty} \mathbb{P}[(\sigma^2,\Upsilon)\in B \mid \bY_1, \dots, \bY_M] = 1 \quad a.s.-\mathbb{P}_{\sigma_0^2,\Upsilon_0},
\end{equation}
where $\Pi_{\sigma^2,\bm{\lambda}\mid\mathcal{T}_0}$ is the conditional distribution of $(\sigma^2,\bm{\lambda})\mid\mathcal{T}=\mathcal{T}_0$ under prior \eqref{eq:bbm_mu_sig_multi}--\eqref{eq:priortree_multi}.
\end{Theorem}

The proof of Theorem~\ref{thm:consist} can be found in the Supplementary Material and exploits the finiteness of the space of tree topologies with $V$ nodes, together with Doob's consistency theorem \citep{doob1949application,ghosal2017fundamentals} and the identifiability results in Section~\ref{indet}. Albeit stated~for~a~single adjacency matrix, these identifiability results  extend to the multiple networks setting.

\begin{figure}[t]
    \centering
    \includegraphics[width = 1.02\linewidth]{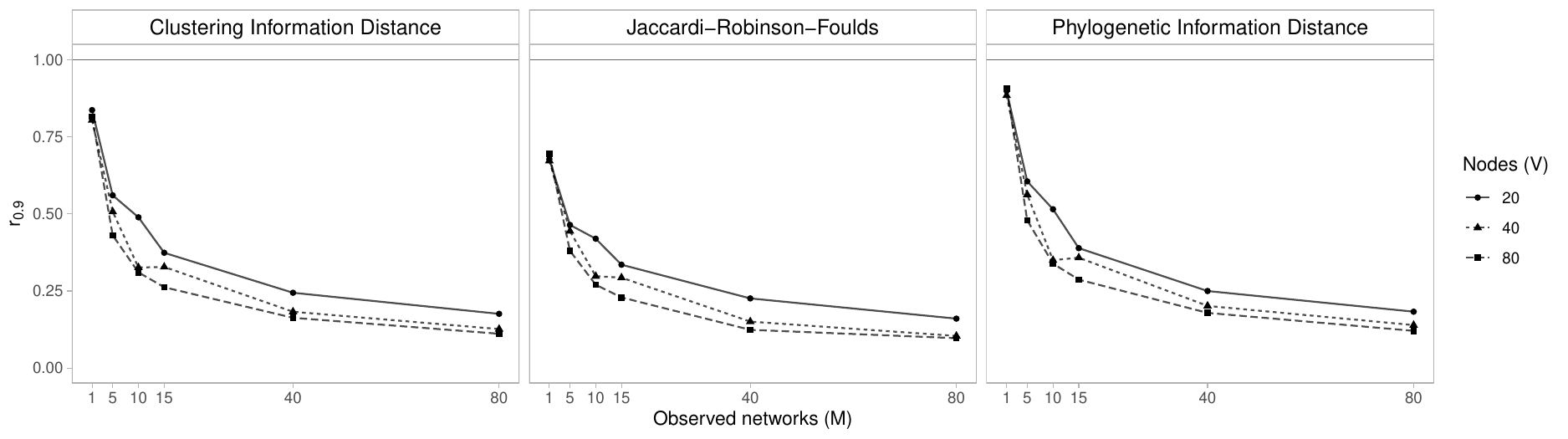}
    \vspace{-10pt}
    \caption{Radius of the $90\%$ credible sets centered at $\Upsilon_0$. This radius is computed under different tree distances and for varying settings of $M =1, 5, 10, 15, 40, 80$ and $V=20,40,80$.}
    \label{fig:consistency}
\end{figure}

Although Theorem~\ref{thm:consist}  provides strong theoretical support to the proposed construction, the consistency result stated  is asymptotic in nature. In practice, $M$ is finite and, hence, it~is~of~interest to assess whether theoretical consistency translates into empirical evidence of effective concentration for the posterior  of $\Upsilon$ around a true tree $\Upsilon_0$. To answer this question, we  simulate networks from model \eqref{eq:lik}--\eqref{eq:bbm_mu_sig_multi} for different settings~of $M =1, 5, 10, 15, 40, 80$~and~$V =20,40,80$,  letting $K=3$ and {\em true} underlying tree architecture $\Upsilon_0$ drawn from prior \eqref{eq:priortree_multi}. 
Conditioned~on~these~networks, we sample multiple trees from the corresponding posterior distribution, as detailed~in~Section~\ref{sec_phy_31}, and then leverage these samples to monitor the concentration of the posterior around $\Upsilon_0$ as $M$ grows, for different settings of $V$, via the radius of the $90\%$ credible sets centered at  $\Upsilon_0$;~see Figure~\ref{fig:consistency}. These~credible~sets contain the $90\%$  sampled trees closest to $\Upsilon_0$ according~to~three~standard notions of normalized tree distance (i.e., the Clustering Information Distance \citep{smith2020information}, the Jaccardi--Robinson--Foulds \citep{bocker2013generalized} and  the Phylogenetic Information Distance \citep{smith2020information},  as implemented in the \texttt{R} package~\texttt{TreeDist} \citep{TreeDist}), whereas the radius coincides with the maximum of the distances between  $\Upsilon_0$  and the trees within~the~$90\%$ credible set. 
As illustrated in Figure~\ref{fig:consistency}, the radius progressively shrinks as $M$ grows, for any~$V$, meaning that the posterior increasingly concentrates around $\Upsilon_0$, thus providing empirical support~in~finite-$M$ settings to Theorem~\ref{thm:consist}. Note that the larger radius at $M=1$ is due to the fact~that,~in this setting, the model has only $M\times K=3$ latent features for each node to learn a complex tree~structure. This translates into higher posterior uncertainty. Such a result does not mean that reasonable~point estimates of the tree cannot be obtained when $M=1$. Rather, it clarifies that the \textsc{phylnet} construction properly quantifies posterior uncertainty.

\section{Posterior Computation and Inference}\label{sec_phy_3}

\setlength\abovedisplayskip{15pt}%
    \setlength\belowdisplayskip{15pt}%

\subsection{Posterior computation via Metropolis-within-Gibbs}\label{sec_phy_31}
Due to the intractability of the posterior distribution induced by model  \eqref{eq:lik}--\eqref{eq:priortree_multi}, we design~a Metropolis-within-Gibbs procedure targeting this posterior. Such a routine samples  iteratively from the full conditional distributions of the model parameters, either via conjugate updates or through Metropolis--Hastings steps with tailored Gaussian proposals, except for $\Upsilon$ which requires suitable tree moves.~Note that, although our overarching focus is on the posterior distribution for $\Upsilon$,~as illustrated in the following, it is practically more convenient to implement a routine targeting~the full posterior induced~by the model  \eqref{eq:lik}--\eqref{eq:priortree_multi},  and then retain only the samples of $\Upsilon$~to~perform~inference on the marginal posterior for the tree as detailed in Section~\ref{sec_phy_32}.

In the following, we detail~the~steps of our algorithm for a generic~$M$; setting $M=1$~yields~directly the sampling strategy for the single-network \textsc{phylnet}  model in \eqref{eq:lik_single}--\eqref{eq:priora}. Focusing first on the scalar~$a$, its full-conditional distribution
\begin{eqnarray*}
\pi(a \mid -) \propto \phi(a; 0, \sigma_a^2)\prod_{m=1}^M\left[\prod_{v=2}^V \prod_{u=1}^{v-1} \frac{(\exp[a\, - \lVert\bz^{(m)}_v - \bz^{(m)}_u \lVert])^{y^{(m)}_{vu}}}{1+\exp[a\, - \lVert\bz^{(m)}_v - \bz^{(m)}_u \lVert]}\right],
\end{eqnarray*}
lacks conjugacy between the Gaussian density $\phi(a; 0, \sigma_a^2)$ and the model likelihood, thereby~requiring a Metropolis--Hastings update. This update relies on a random walk Gaussian proposal with~variance tuned as in  \citet{andrieu2008tutorial} to target the ideal acceptance rate of $\bar{\alpha}=0.23$, which ensures an effective balance  between local and long-distance exploration.~In~practice, this is achieved by defining the standard deviation \smash{$\eta_a^{(s)}$} of the Gaussian proposal at the step $s$  as \smash{$\log \eta_a^{(s)}=\log \eta_a^{(s-1)}+s^{-0.8}(\alpha^{(s)}-\bar{\alpha})$}, where \smash{$\alpha^{(s)}$} is the acceptance probability~at~iteration~$s$. 

The above tuning step is implemented not only for  $a$, but also for all the quantities in model~\eqref{eq:lik}--\eqref{eq:priortree_multi}  sampled via Metropolis--Hastings based on Gaussian proposals.~This~is~the~case of the latent~features encoded in the $K \times V$ matrix $\bZ^{(m)}$, for $m=1, \ldots, M$, whose full-conditional is
\begin{eqnarray*}
\begin{split}
\pi(\bZ^{(m)} \mid -) \propto \left[\prod_{k=1}^K \phi_V(\bZ^{(m)\intercal}_{[k]}; \mu^{(m)}_k\bone_V,\sigma^2\bSigma_\Upsilon)\right]\cdot\left[\prod_{v=2}^V \prod_{u=1}^{v-1} \frac{(\exp[a\, - \lVert\bz^{(m)}_v - \bz^{(m)}_u \lVert])^{y^{(m)}_{vu}}}{1+\exp[a\, - \lVert\bz^{(m)}_v - \bz^{(m)}_u \lVert]}\right],
\end{split}
\end{eqnarray*}
for each $m=1, \ldots, M$. Crucially, this form allows to implement parallel Metropolis--Hastings~updates for \smash{$\pi(\bZ^{(m)}\mid -)$} over $m=1,\dots,M$, with each parallel update sampling jointly from~the~$K$-dimensional features vector \smash{$\bz^{(m)}_v$} of every node $v$, for $v=1, \ldots, V$.

Given the above samples, the rate \smash{$\sigma^2$} and the centering parameters \smash{$\mu^{(m)}_k \in \mathbb{R}$}, for $k=1, \ldots, K$,~$m=1,\ldots, M$, admit conjugate inverse-Gamma and Gaussian full-conditionals, respectively. More specifically, let $\alpha^*_{\sigma}=\alpha_\sigma +VKM/2$ and $\beta^*_{\sigma}=\beta_\sigma + \,\bar{\bz}^\intercal \left( \bI_{KM}\otimes \bSigma_\Upsilon^{-1} \right) \bar{\bz}/2$, then
\begin{eqnarray*}
\begin{split}
&\pi(\sigma^{-2} \mid -) = \frac{ (\beta^*_{\sigma})^{\alpha^*_{\sigma}}}{\Gamma(\alpha^*_{\sigma})} (\sigma^{-2})^{\alpha^*_{\sigma}-1} \exp(-\beta^*_{\sigma}\sigma^{-2}), \quad \mbox{and}\\
&\pi(\mu^{(m)}_k \mid -)= \phi\bigg(\mu^{(m)}_k; \frac{\bone^{\intercal}_V(\sigma^{2}\bSigma_\Upsilon)^{-1}\bZ^{(m)\intercal}_{[k]}}{\sigma_\mu^{-2}+\bone^{\intercal}_V(\sigma^{2}\bSigma_\Upsilon)^{-1} \bone_V}, \frac{1}{\sigma_\mu^{-2}+\bone^{\intercal}_V(\sigma^{2}\bSigma_\Upsilon)^{-1} \bone_V}\bigg),
\end{split}
\end{eqnarray*}
for  every $m=1,\dots,M$ and $k=1,\dots,K$, where   $ \bar{\bz} = \smash{(\bar{\bz}^{(1)\intercal}_1,\dots,\bar{\bz}^{(1)\intercal}_K,\dots,\bar{\bz}^{(M)\intercal}_1,\dots,\bar{\bz}^{(M)\intercal}_K)^\intercal} \in \R^{V K M}$,  with \smash{$\bar{\bz}^{(m)}_k=\bZ^{(m)\intercal}_{[k]} - \mu^{(m)}_k\bone_V$}. Adapting popular routines for latent space model  \citep[e.g.,][]{krivitsky2009representing}, the above updates are combined with a subsequent~Metropolis--Hastings step~re-sampling~the quantities \smash{$(\bZ^{(1)},\dots,\bZ^{(M)},\bmu^{(1)},\dots,\bmu^{(M)},\sigma^2)$} jointly in order to~improve mixing. This  is accomplished by proposing a~modified version of the sampled values~for $\bZ^{(1)},\dots,\bZ^{(M)}$, $\bmu^{(1)},\dots,\bmu^{(M)}$ and $\sigma^2$, jointly rescaled by $h \sim \mbox{N}(0,\sigma^2_{h})$, and then accepting~or rejecting~the~proposed values based on the corresponding Metropolis--Hastings acceptance probability. 

In order to sample from the full-conditional of  $\Upsilon$, we adapt strategies available in the literature on Bayesian phylogenetic trees \citep[e.g.,][]{kelly2023lagged,chen2014bayesian} treating the sampled features  $\bZ^{(1)}, \ldots, \bZ^{(M)}$ as observed traits. We rely, in particular, on Metropolis--Hastings~updates based on five symmetric moves that ensure ergodicity of the chains in the tree space:

\vspace{-6pt}
\begin{itemize}
\item \texttt{Tips interchange}: for each leaf node randomly select another leaf and swap them;
\vspace{-5pt}
\item \texttt{Subtree exchange}: randomly select two disjoint subtrees and swap them;
\vspace{-5pt}
\item \texttt{Tree-node age move}: randomly select an internal node and shift its age, corresponding~to~the~expansion or contraction of the branches connecting the selected node, its parent in the tree and~the child nodes, keeping unchanged the total height of the tree;
\vspace{-5pt}
\item \texttt{Subtree pruning and regrafting (\textsc{spr})}: randomly select a subtree, prune it and re-attach~the resulting subtree to another suitable position in the tree;
\vspace{-5pt}
\item \texttt{Local-\textsc{spr}}: \textsc{spr} move with the restriction of regrafting in any suitable branch of the subtree~rooted at the parent node of the pruned subtree.
\end{itemize}

To enhance mixing and convergence, we follow recommended practice in Bayesian phylogenetics and implement all the above moves for a pre-specified number of times in each step~of~the~Gibbs~routine, while checking that the proposed move does not  violate the ultrametric property. 

Finally, given the sample of the tree $\Upsilon$, the associated prior hyperparameter $b$ is drawn~via~a~random walk Metropolis--Hastings update from the full conditional
\begin{eqnarray*}
\pi(b \mid -) \propto \frac{\beta_b^{\alpha_b}}{\Gamma(\alpha_b)} \frac{\exp(-\beta_b/b)}{b^{\alpha_b+1}} \pi_{\textsc{bdt}}(\Upsilon; b,0),
\end{eqnarray*}
where $\pi_{\textsc{bdt}}(\Upsilon; b,0)$ denotes the density of a birth and death process with rates $b$ and $0$ computed at the tree $\Upsilon$. In this case we propose $\log b$ again from  a Gaussian. 

The above Metropolis-within-Gibbs is initialized with random draws from the prior~for~all~the parameters, except for $a$ whose starting value is set at the median of the \textsc{mle} estimates provided by $M$ different latent space models \citep{hoff2002latent} fitted separately via the \texttt{R}-package \texttt{latentnet} \citep{latentnet} to each observed network $\bY^{(1)}, \ldots, \bY^{(M)}$. The sampling steps discussed above are performed in random order within  each full cycle of the Metropolis-within-Gibbs routine. In our experience, this choice helped to further improve mixing.

\subsection{Posterior inference on $\Upsilon$}\label{sec_phy_32}
As mentioned in Section~\ref{sec_1},  the architecture of the tree $\Upsilon$ allows us to unveil increasingly-nested~modular hierarchies among nodes that inform on multiscale network structures often observed in practice. This motivates our overarching focus on posterior inference for such an architecture, via the posterior~samples for  $\Upsilon$ produced by the Metropolis-within-Gibbs outlined in Section~\ref{sec_phy_31}. 

Our aim is to obtain not only a point estimate of $\Upsilon$, but to fully quantify posterior uncertainty. To this end, we rely on  two standard summaries of the sample of trees: {\em DensiTrees} and~{\em consensus trees} (see Figures \ref{figure:crime_2} and  \ref{figure:brain_2}).
A {\em DensiTree} \citep{bouckaert2010densitree} allows full visualization of~the posterior uncertainty, by overlaying all  trees in the sample under a common ordering~of~the~leaves. This~is~a~useful visualization, but can sometimes be so blurry as to be difficult to interpret.~A~{\em consensus~tree},~on~the~other hand, summarizes the posterior via a single tree.~This~is~accomplished by displaying only the splits that appear within the sampled trees  in a proportion above a threshold $p$. Hence, some  internal~nodes of the consensus~tree~may~be~multifurcating when no split has high enough empirical posterior probability, thereby facilitating the identification of those parts of the~tree  that are uncertain or poorly-resolved in the posterior. The branch~lengths~of the consensus trees are  computed as the mean of the lengths~of~the~corresponding branches~in~the~tree~samples.

See Figures \ref{figure:crime_2} and  \ref{figure:brain_2} for a graphical representation of the {\em DensiTrees} and  {\em consensus trees} resulting from the application of the  \textsc{phylnet} model  to networks from criminology and neuroscience.

\section{Simulation Studies}\label{sec_phy_4}
We  assess the performance of the proposed  \textsc{phylnet} model and  its empirical improvements~over~alternative competitors in two simulation scenarios that exhibit different network structures. 
In our first simulation scenario (see Figure~\ref{figure:simu_1}), there is no multiscale pattern. We~consider, in particular, $M=10$ networks having $V=80$ nodes whose connections are sampled from independent Bernoulli~variables with edge probabilities displaying community-type architectures between 5 groups, with no evidence of multiscale patterns. 
In the second simulation scenario, we consider  instead $M=30$~networks of size $V=60$  simulated under the data generative process of the  \textsc{phylnet} model discussed in Section~\ref{sec_phy_23}. More specifically, we first generate the true underlying tree $\Upsilon_0$ from~the~prior in \eqref{eq:priortree_multi} with $b_0=0.6$. Conditioned on this tree, the node-specific latent features are simulated,~across the $M=30$ networks, from \eqref{eq:bbm_row_multi}, setting \smash{$\sigma_0^2=0.6$ and $\mu^{(m)}_{k,0}=0$}, for all $m=1, \ldots, 30$ and $k=1, \ldots, 3$. Finally,~given these simulated features, the edges within the  adjacency matrices $\bY_1, \ldots, \bY_{30}$ are generated from the latent space model in \eqref{eq:lik}, with $a_0=2.6$. Figure~\ref{figure:simu_2} provides a graphical illustration of $\Upsilon_0$,~along~with~the~expectation of the $M$  edge-probability~matrices and examples of four simulated adjacency matrices~in~the second scenario.~Unlike~the first scenario, in this case~the~nodes~display  heterogeneous~and~more~nuanced multiscale patterns, thus allowing to assess  the \textsc{phylnet} construction  both~in situations where networks are not simulated from the generative process underlying the proposed model  (first scenario), and~also~in~correctly-specified, yet challenging, settings (second~scenario).   

\begin{figure}[t!]
\centering
    \includegraphics[trim=0cm 0cm 0cm 0cm,clip,width=0.95\textwidth]{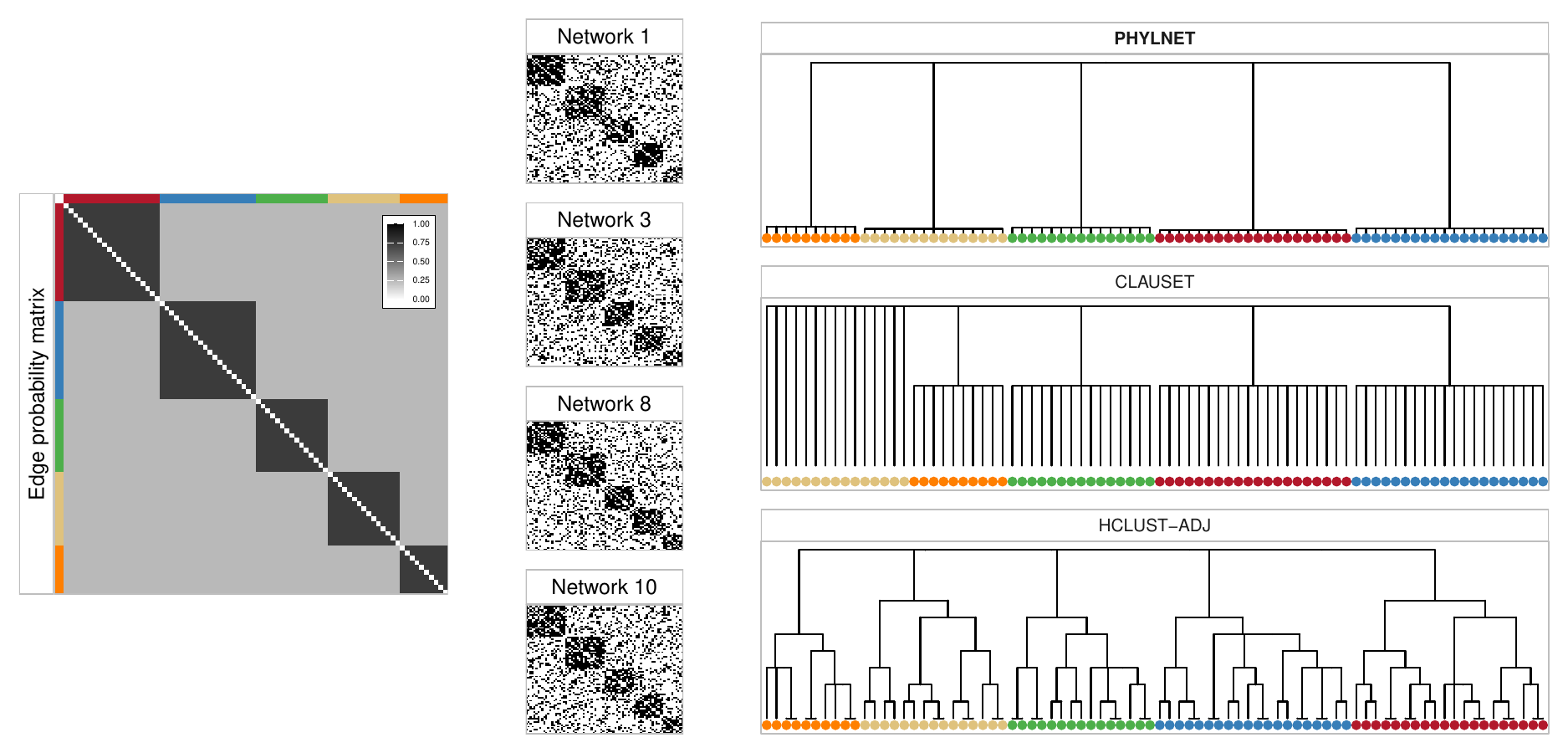}
    \caption{First simulation scenario: community-type structure. From left to right: matrix of edge probabilities;~four~examples of simulated networks; inferred trees under the   \textsc{phylnet} model and alternative competitors.~Colors~indicate~community membership. For the \textsc{phylnet} model, the {\em consensus tree} is based on a 0.8 threshold proportion.}
    \label{figure:simu_1}
\end{figure}

Posterior inference under the \textsc{phylnet} model proceeds via the Metropolis-within-Gibbs outlined in Section~\ref{sec_phy_31}, setting  $\alpha_b = \alpha_\sigma = \beta_b = \beta_\sigma = 1$, $\sigma_a=10$ and \smash{$\sigma_{\mu}=\sqrt{1000}$}. Although~in~our experiments the  \textsc{phylnet}  model proved robust to moderate deviations of such hyperparameter settings, these~default values have always led to sensible results when applied to substantially~different~networks~in~both~simulations and applications. Regarding the choice of $K$, we implemented separate latent space models~as in \citet{hoff2002latent} (using the \texttt{latentnet} \texttt{R}--package) for each of the simulated networks in both scenarios to explore how different settings for $K$ were able to characterize the structure of~the observed networks. This procedure led to select $K=3$, which was shown to achieve a sensible balance between dimensionality reduction and goodness-of-fit. For both scenarios, posterior inference relies on 4 chains of $10^5$ iterations, with a burn-in window of 75{,}000 and a thinning every~50 iterations, as suggested by \textsc{mcmc} diagnostics. Note that these conservative \textsc{mcmc}  settings are common in Bayesian phylogenetic inference, due to the complex topology of tree spaces  \citep[e.g.,][]{chen2014bayesian}.

\begin{figure}[t!]
\centering
    \includegraphics[trim=0cm 0cm 0cm 0cm,clip,width=0.95\textwidth]{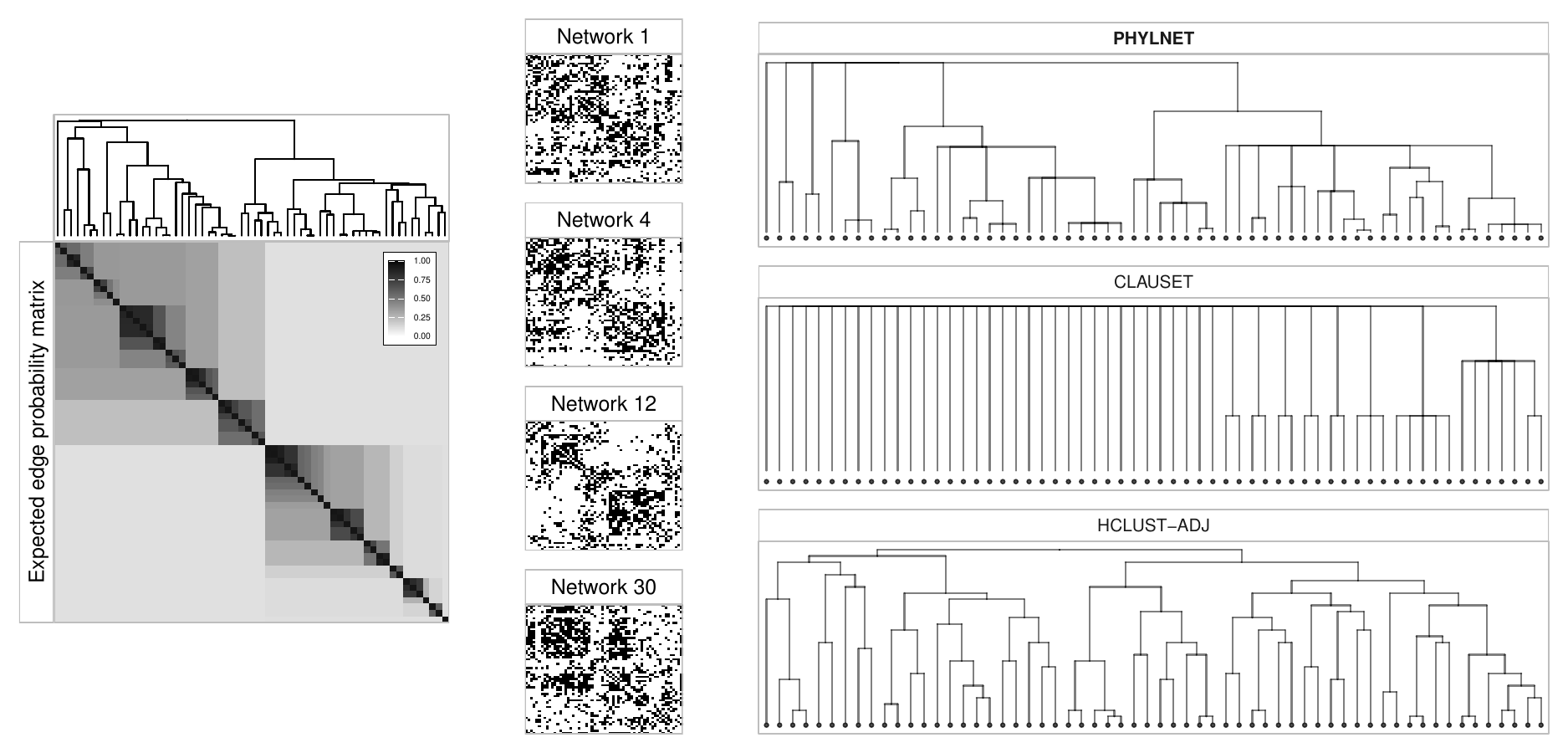}
    \caption{Second simulation scenario: tree-type multiscale structure. From left to right: True tree $\Upsilon_0$ and expectation~of the $M$  edge-probability matrices; four examples of simulated networks; inferred trees under the proposed  \textsc{phylnet} model and alternative competitors. For the \textsc{phylnet} model, the {\em consensus tree} is based on a 0.8 threshold proportion.}
    \label{figure:simu_2}
\end{figure}

Figures~\ref{figure:simu_1}--\ref{figure:simu_2} display the {\em consensus trees} (see Section~\ref{sec_phy_32}) obtained under the   \textsc{phylnet} model, for both the first and second simulation scenarios, respectively. To further clarify~the~empirical gains~of  \textsc{phylnet}, these architectures are also compared with alternative tree-based reconstructions~of~the~multiscale patterns underlying the simulated networks.  As mentioned in Section~\ref{sec_1}, current~literature comprises some generalizations of stochastic block models to learn tree-structured node hierarchies \citep[e.g.,][]{clauset2008hierarchical}, but lacks solutions that can infer heterogeneous tree-based generative mechanisms in latent space models. For this reason, besides the model~by \citet{clauset2008hierarchical} (\textsc{clauset}), we consider as additional competitor the dendrogram resulting from hierarchical clustering applied to  the simulated adjacency matrices (\textsc{hclust-adj}). This tree-based topology has potential to accommodate heterogeneity at the node level and learn flexible hierarchies.  As such, it provides a natural heuristic benchmark aligned with the original motivations behind the proposed \textsc{phylnet} model. Notice that the model by \citet{clauset2008hierarchical}  and hierarchical clustering are originally designed for a single adjacency matrix. Hence, both  \textsc{clauset} and \textsc{hclust-adj} are applied to the weighted  network~that~arises from the sum of the $M$ binary adjacency matrices. The model by \citet{clauset2008hierarchical} characterizes the edge weights under a binomial likelihood, whereas hierarchical clustering leverages the resulting sum as a similarity matrix. An alternative to this aggregation would be to apply \textsc{clauset} and \textsc{hclust-adj}  separately to each adjacency matrix and then summarize the $M$ inferred node hierarchies through a shared  {\em consensus tree}. Although such a solution aligns with   \textsc{phylnet} in allowing increased flexibility  across networks, the resulting trees displayed lower accuracy relative to those inferred under~the~aggregated~perspective.

The visual comparison in Figures~\ref{figure:simu_1}--\ref{figure:simu_2} among the tree inferred by the \textsc{phylnet} model and those obtained under the aforementioned competitors highlights the noticeable empirical gains achieved by \textsc{phylnet}  in learning complex tree-based structures behind different networks. Although the networks in the first scenario are not simulated under the \textsc{phylnet} data-generative mechanism, it is sensible to expect that a suitable tree-based representation for the corresponding block structures would combine a root multifurcation into the five different communities with additional nested multifurcations~of the nodes within each of these communities, compressed at the terminal~leaves. Such a compression would correctly indicate that nodes within each community are stochastically equivalent \citep[e.g.,][]{nowicki2001estimation}, and hence indistinguishable in terms of connectivity behavior.~While~all~the methods under analysis properly learn the node communities in the first scenario, as is clear from Figure~\ref{figure:simu_1}, \textsc{phylnet} aligns more closely with this expected tree-based representation of  community~structures. The empirical improvement in learning nested modular hierarchies among the nodes is evident~also~in Figure~\ref{figure:simu_2} when comparing the consensus tree inferred by \textsc{phylnet} with the true $\Upsilon_0$ underlying the simulation of the multiple networks in the second scenario. Table~\ref{tab:dist_true_tree}~quantifies~these~gains~via~the~distances between the inferred~trees~and the true $\Upsilon_0$, under the metrics discussed in Section~\ref{sec_consis} (Jaccardi--Robinson--Foulds \citep{bocker2013generalized}; Clustering Information Distance  \citep{smith2020information},  and  Phylogenetic Information Distance \citep{smith2020information}). More specifically, we compute the distance between  $\Upsilon_0$ and its point estimate under the non-Bayesian method (\textsc{hclust-adj}), while for the Bayesian solutions (\textsc{phylnet} and  \textsc{clauset})~we consider a more in-depth assessment of posterior concentration.~This is accomplished by computing the distance of $\Upsilon_0$ from~each posterior sample of $\Upsilon$, and then displaying the average of the resulting distances  along with the associated quantiles~0.05~and~0.95.

\begin{table}[t]
\centering
\small
\caption{For the different methods under analysis, distances between the inferred trees and the true $\Upsilon_0$~in~the~second scenario, under the tree metrics discussed in Section~\ref{sec_consis} (Jaccardi--Robinson--Foulds (\textsc{jrf}) \citep{bocker2013generalized}; Clustering Information Distance (\textsc{cid}) \citep{smith2020information},  and  Phylogenetic Information Distance (\textsc{pid}) \citep{smith2020information}). For the non-Bayesian method (\textsc{hclust-adj}) we report the distance between  $\Upsilon_0$ and its point estimate.~For~the~Bayesian methods (\textsc{phylat} and \textsc{clauset}), we compute these distances for each posterior sample of $\Upsilon$ and display the resulting empirical average  along with the quantiles  0.05 and 0.95 in brackets.} 
\scalebox{1}{
\begin{tabular}{l|r|r|r}
Distance &  \textsc{phylnet} & \textsc{clauset} &  \textsc{hclust-adj}\\ 
  \hline
\textsc{jrf} & \textbf{0.18 [0.13,0.21]} & 0.44 [0.36,0.50] & 0.34 \\ 
  \textsc{cid} & \textbf{0.22 [0.16,0.26]} & 0.58 [0.45,0.63] &  0.43 \\ 
  \textsc{pid} & \textbf{0.22 [0.17,0.27]} & 0.61 [0.49,0.67] &  0.47 \\ 
   \hline
\end{tabular}}
\label{tab:dist_true_tree}
\end{table}

The results in Table~\ref{tab:dist_true_tree} confirm the superior performance of \textsc{phylnet} in achieving substantially-improved posterior concentration around the true  $\Upsilon_0$. This seems not the case for the model~by~\citet{clauset2008hierarchical} (\textsc{clauset}), whose reduced performance is arguably attributable~to~the~fact~that \textsc{clauset} was originally developed as a tree-based extension of stochastic block models,~and~hence fails to characterize more heterogeneous node-specific~patterns. This is further evident in the slight performance gains displayed by \textsc{hclust-adj} which, unlike~for  \textsc{clauset}, has potential to accommodate more flexible tree-based structures accounting for heterogeneity at the node level. However, this heuristic method is still significantly less accurate than \textsc{phylnet}, achieving about half of the accuracy of the model we propose, across all considered metrics. Such a result~motivates two comments. First, this gain is reminiscent of a well-known effect in phylogenetics, where~evolution-based models tend to improve upon distance-based methods such as \textsc{upgma} or neighbor-joining \citep{felsenstein1981evolutionary, hillis1994hobgoblin, kuhner1994}. Second,  this heuristic does not benefit from a fully-Bayesian specification leveraging latent space~representations. 

We find that \textsc{phylnet} also reconstructs correctly the scalar parameters of the model. All the $90\%$ credible intervals cover the true value ($a_0=2.6$, CI $[2.51,2.64]$; $\sigma_0^2=0.6$; CI $[0.49,0.62]$; $b_0=0.6$, CI $[0.43,0.70]$). This result, combined with the structured borrowing of information among~node~features achieved by the flexible tree, also allows to obtain accurate estimates~of the edge probabilities underlying the $M$ adjacency matrices. In particular, the root mean squared~error~between  the true \smash{$\theta^{(m)}_{vu}$} and the corresponding posterior expectation obtained via  \textsc{phylnet}~under \eqref{eq:lik} (for each $1 \leq u < v \leq V$ and $m=1, \ldots, M$) is $0.11$. This improves the error of $0.13$ resulting from separate Bayesian latent space models \citep{hoff2002latent} applied separately to the $M$ adjacency matrices under the default settings in  \texttt{latentnet}. Such a result demonstrates that structured representations of feature formation mechanisms not only expand inferential capabilities, but also improve recovery of the features themselves, and, consequently, of edge probabilities.

Having established the robustness and reliability of \textsc{phylnet} on synthetic data, we now turn~to~two relevant applications from criminology and neuroscience, and illustrate further  the ability of  \textsc{phylnet}  to learn tree-based node architectures underlying real-world network data.

\section{Applications}\label{sec_phy_5}

\subsection{Criminal networks}\label{sec_phy_51}

There is a substantial interest in criminology on unveiling the  organizational structures~of~criminal networks from the analysis of the complex connectivity patterns among the corresponding members \citep[see, e.g.,][]{calderoni2017communities,coutinho2020multilevel,campana2022studying}. Current attempts to address this  goal are  constrained by the lack of network models capable of learning~the complex multiscale patterns underlying structured criminal organizations,~and~by~the~challenges that arise from combining the multiple noisy data sources available (including separate investigations from different law-enforcement agencies) \citep[e.g.,][]{bright2022reprint,diviak2022key}.

\begin{figure}[b]
      \vspace{-5pt}
\centering
    \includegraphics[trim=0cm 0cm 0cm 0cm,clip,width=0.84\textwidth]{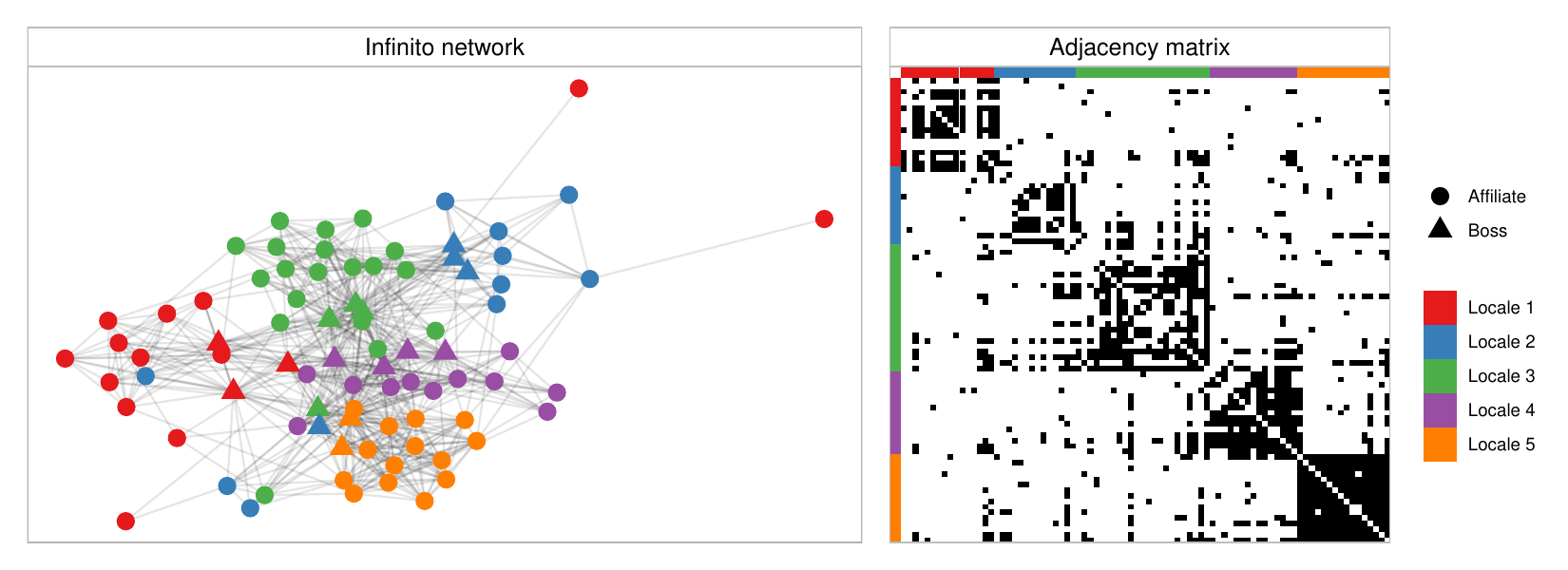}
      \vspace{-5pt}
    \caption{Example of one network in the criminology application. Left: graphical representation of the network,~where the positions of the different criminals (nodes)  are obtained via force directed placement \citep{fruchterman1991graph}, whereas colors and shape denote {\em locali} membership and role, respectively. In 'Ndrangheta, {\em locali} correspond to subgroups in the criminal organization that administer crime in specific territories. Right: Adjacency~matrix~representation of the network.  White and black colors for the entries of the matrix denote non-edges and edges respectively.}
    \label{figure:crime_1}
    \vspace{-10pt}
\end{figure}

Motivated by this gap, we consider the \textsc{phylnet} model within a hypothetical, yet realistic,~scenario based on data from a large law enforcement operation  conducted in Italy from $2007$~to~$2009$~for~monitoring and then disrupting a highly-structured  'Ndrangheta mafia organization~in~the~area of Milan. To this end, we consider the dataset studied recently~by~\citet{legramanti2022extended} and \citet{lu2024zero} which comprises information on the participation~of~$V=84$ criminals to $47$ monitored summits of the criminal organization, as reported in the judicial documents. \citet{legramanti2022extended} and \citet{lu2024zero} map this information into a single network with edges denoting either dichotomized  \citep{legramanti2022extended} or weighted  \citep{lu2024zero} co-attendances among each pair of criminals to the summits. Consistent with the previous considerations,~we~consider instead the hypothetical scenario in which each of $M=10$ law-enforcement agencies has monitored a different subset of $35$ summits~sampled~at random and  with replacement from the total of $47$. This yields $M=10$ adjacency~matrices~with entries \smash{$y^{(m)}_{vu}=y^{(m)}_{vu}=1$} if criminals $v$ and $u$ co-attended at least one of the summits monitored by agency $m$, and \smash{$y^{(m)}_{vu}=y^{(m)}_{vu}=0$} otherwise, for  $1\leq u < v \leq V$ and $m=1, \ldots, M$. Since criminal network data are often prone to data quality issues  \citep[e.g.,][]{diviak2022key}, we also add a level of contamination to each adjacency matrix by~flipping~$2\%$ of its edges.  Figure \ref{figure:crime_1} illustrates graphically one of these adjacency matrices, which we jointly model under the  \textsc{phylnet} construction. Although the~scenario~explored~is hypothetical, this choice is useful in clarifying how separate investigations subject to data incompleteness  and contamination, could be eventually pooled under  \textsc{phylnet} to obtain an informative tree-based reconstruction of the organizational structure of the monitored Mafia group.

\begin{figure}[t]
\centering
    \includegraphics[trim=0cm 0cm 0cm 0cm,clip,width=0.96\textwidth]{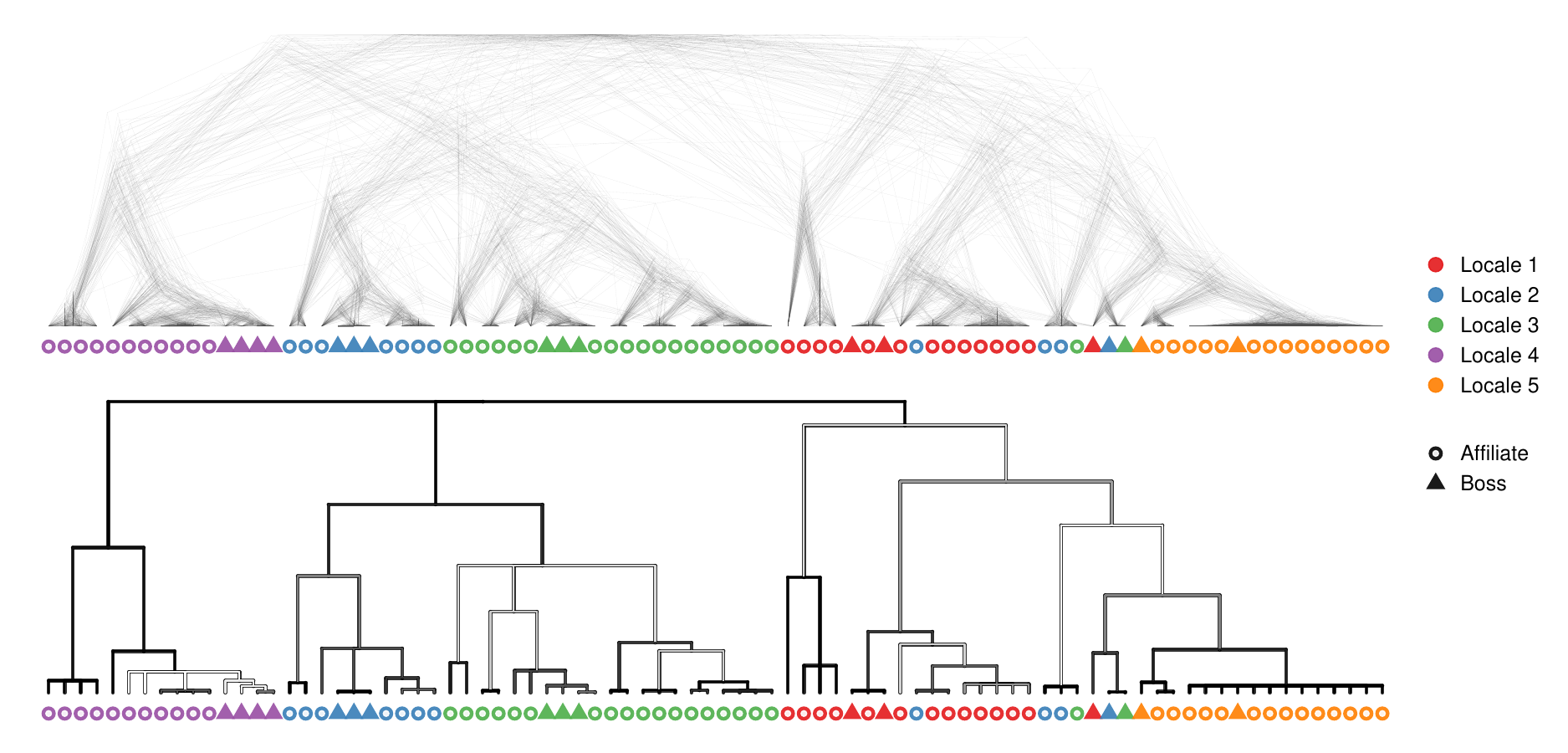}
    \caption{{\em DensiTree} and {\em consensus tree} summarizing the posterior of the tree in the criminal~networks~application. In the  {\em consensus tree}, branch colors represent the posterior support of the split rooted at the parent node, ranging from the threshold level 0.6 (white) to 1.0 (black); when no split has posterior support greater than 0.6, the tree is multifurcating.}
    \label{figure:crime_2}
\end{figure}

Figure \ref{figure:crime_2} illustrates the output of this analysis via the {\em DensiTree} and {\em consensus tree} learned~by the    \textsc{phylnet} model with the same hyperparameters and \textsc{mcmc} settings as in the simulation~studies in Section~\ref{sec_phy_4}. Interestingly, these trees reveal unexplored hierarchies of increasingly-nested macro, meso and micro modules attributable to {\em locali} differentiations, role specialization within {\em locali}, and blood-family ties, respectively. Besides providing novel quantitative support to criminology theories on the organizational structure of  'Ndrangheta \citep[e.g.,][]{paoli2007mafia,catino2014mafias,calderoni2017communities}, this  tree structure also unveils specific criminals with nuanced positions that partially depart from those of the more rigid 'Ndrangheta architecture. This is the case, for example,~of~the two affiliates in the clades of bosses from the blue~and~green {\em locali}. According to the judicial~acts, these affiliates are high-rank members with a key role in overseeing~implementation of bosses' decisions. Consistent~with~the~learned trees, this role is closer to the one covered by a boss than~an affiliate. It is also interesting to note three bosses of different {\em locali} form a subtree with~the~clade corresponding to the orange {\em locale}. Although this positioning might appear unusual, in the judicial acts these bosses are listed among those supporting a failed attempt to increase the independence of the 'Ndrangheta group in the area of Milan from the leading families in Calabria. The inferred tree suggests that this event resulted in a need for these~bosses to move away from the corresponding {\em locale} and move their area of influence towards the orange one. This movement is also followed by the close affiliates. These results clarify the ability of   \textsc{phylnet} to learn not only higher-level modular hierarchies within the network, but also more nuanced lower-level heterogeneous behaviours of  specific nodes, thereby achieving an effective balance between the rigid high-level organizational structure of 'Ndrangheta and the more fluid local~positioning~of~its~members.

\subsection{Brain networks}\label{sec_phy_52}
State-of-the-art studies in neuroscience highlight the presence of symmetries and multiscale modular patterns in structural brain connectivity networks \citep{bullmore2009complex,meunier2010modular,rubinov2010complex,bullmore2012economy,betzel2017multi,esfahlani2021modularity}. These findings are supported by  anatomical considerations~and empirical evidences from the analysis of replicated structural brain network data via community detection algorithms or stochastic block models \citep[e.g.,][]{meunier2010modular}. However,~although useful for identifying brain modules, both community detection  and stochastic block models are not originally designed for inferring increasingly-nested hierarchies and symmetries among  modules from the observed multiscale structures in human brain networks. As a result, current quantitative analyses might provide an overly-coarsened reconstruction of brain organization that fails to characterize complex modules at different scales along with heterogeneous patterns at the~level~of~brain~regions.

\begin{figure}[b]
\centering
    \includegraphics[trim=0cm 0cm 0cm 0cm,clip,width=0.84\textwidth]{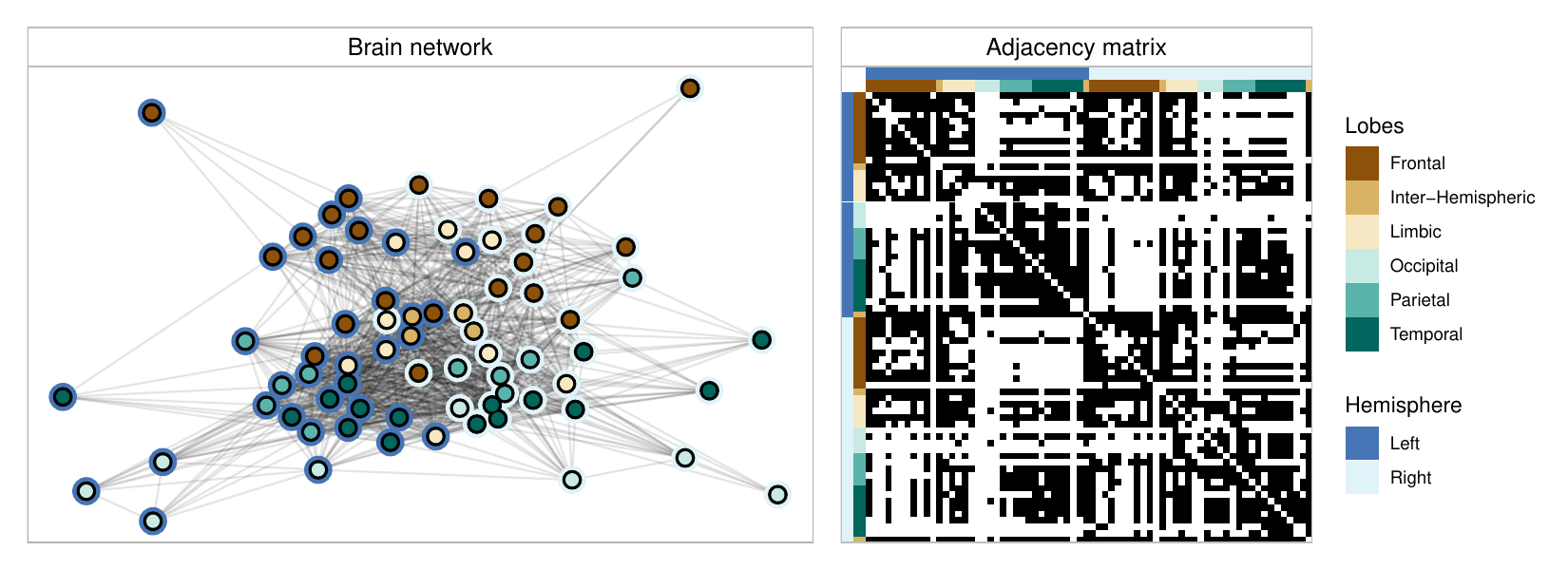}
\vspace{-5pt}
    \caption{Examples of one network in the neuroscience application. Left: graphical representation of  a network,  where the positions of the different brain regions (nodes) are obtained via force directed placement \citep{fruchterman1991graph}, whereas colors denote lobe and hemisphere membership, respectively. Right: Adjacency matrix representation of the network.  White and black colors for the entries of the matrix denote non-edges and~edges respectively.}
    \label{figure:brain_1}
    \vspace{-5pt}
\end{figure}

To overcome these limitations, we apply the  \textsc{phylnet} model to structural brain connectivity~networks from the Enhanced Nathan Kline Institute Rockland Sample project.~The~pre-processed~data are available at \url{https://neurodata.io/mri/} and comprise white matter fibers connectivity measures among $V=68$ anatomical regions from the Desikan atlas \citep[e.g.,][]{desikan2006automated} for $M=20$ individuals whose brain has been scanned twice via diffusion tensor imaging (\textsc{dti}). We represent these data as $M=20$ binary adjacency matrices with generic entry \smash{$y^{(m)}_{vu}=y^{(m)}_{vu}=1$} if at least one white matter fiber is recorded among brain regions $v$ and $u$ in the two \textsc{dti} scans of individual $m$, and \smash{$y^{(m)}_{vu}=y^{(m)}_{vu}=0$} otherwise, for  $1\leq u < v \leq V$, $ m=1, \ldots, M$; see Figure~\ref{figure:brain_1} for a representative example of the resulting adjacency matrices.

\begin{figure}[t]
\centering
    \includegraphics[trim=0cm 0cm 0cm 0cm,clip,width=0.96\textwidth]{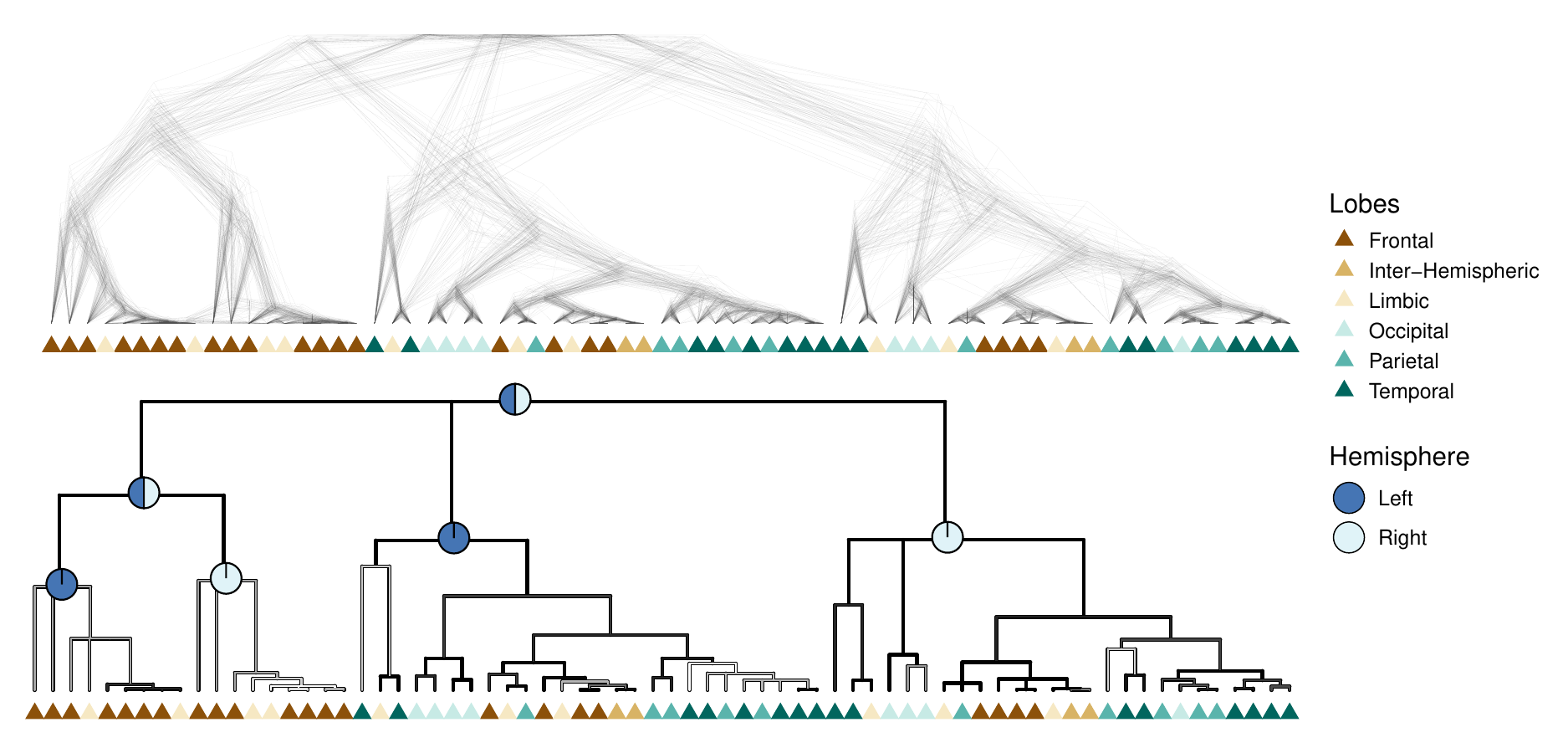}
    \caption{{\em DensiTree} and {\em consensus tree} summarizing the posterior of the tree in the brain networks application. In the {\em consensus tree}, branch colors represent the posterior support of the split rooted at the parent node, ranging from the threshold level 0.6 (white) to 1.0 (black); when no split has posterior support greater than 0.6,~the~tree is multifurcating.}
    \label{figure:brain_2}
\end{figure}

Figure~\ref{figure:brain_2} displays the {\em DensiTree}  and  {\em consensus tree}  inferred by the  \textsc{phylnet} model under~the same hyperparameters' specification and \textsc{mcmc} settings as in the simulation studies in Section~\ref{sec_phy_4}. These trees unveil a previously-unexplored representation of the nested symmetries and hierarchies in brain organization that points toward a macro-level frontal-backward partition followed by an intermediate nested division according to the two hemispheres and a local-level organization which is generally coherent with lobes classification and crucially respect bilateral symmetries for those pairs~of~nodes~characterizing the same region in the two different hemispheres. Interestingly, these symmetries are preserved also for those regions whose inferred position in the tree departs from lobe classification.~For example, the two left and right regions from the limbic lobe that are allocated to the frontal clade (i.e., {\em  rostralanteriorcingulate} and {\em  caudalanteriorcingulate}) are anatomically closer to the anterior part of the human brain than other regions from the frontal lobe, such as the {\em  paracentral}, {\em precentral}, {\em caudalmiddlefrontal} and part of the {\em superiorfrontal} which, in fact, are assigned to the backward clade under~the  \textsc{phylnet} model. Within the backward clade it is also interesting to notice three brain regions (i.e.,~{\em temporalpole}, {\em  entorhinal}, {\em parahippocampal}) that form a peculiar branch present in both~the~left~and~right~hemisphere division. Besides being spatially closer, these regions are conjectured~to~be~closely-related both anatomically and functionally \citep{blaizot2010human}. The inferred tree in Figure~\ref{figure:brain_2} seems to provide empirical evidence in favor of this conjecture, while also aligning with recent analyses of brain modules that point toward the macro-level  frontal-backward, rather than left-right, division of brain networks \citep{esfahlani2021modularity}. In addition to these analyses, the tree inferred within Figure~\ref{figure:brain_2}~clarifies that the modular hierarchies and symmetries in human brains~are~more~complex~and~nuanced~than current classifications of brain regions into lobes, thus motivating future research on the definition~of anatomical brain parcellations and taxonomies. These results are further strengthened by the {\em DensiTree} in Figure~\ref{figure:brain_2} that showcases a relatively low variability of the posterior for the phylogenetic tree, meaning that the inferred structures are generally stable and~coherent across the studied individuals.

\vspace{-5pt}
\section{Conclusions}\label{sec_phy_6}
Although there is evidence from several fields that real networks display multiscale structures and hierarchical organization \citep[][]{ravasz2003hierarchical,gosztolai2021unfolding},~state-of-the-art latent space models lack extensions capable of incorporating and learning these recurring architectures. The  \textsc{phylnet}  model  addresses this gap via a structured representation~of~the~generative process underlying node latent features that creates a unique bridge with Bayesian phylogenetics.  The empirical results in Sections~\ref{sec_phy_4} and \ref{sec_phy_5} show that this perspective can reconstruct~informative and yet-unexplored tree-based representations of complex network organization in several settings displaying substantially different multiscale connectivity.

The above advancement motivates several directions for future research. An important one is to extend the  model from binary undirected settings to directed and weighted network contexts. Both extensions require straightforward modifications of our construction. In the former case it suffices to consider two vectors of latent features for each node regulating~its~sender and receiver behavior, and then learn two separate trees characterizing the formation process of these sender and receiver features, respectively (a related perspective can be employed also to address bipartite network settings, which include recommender systems). Inclusion of weighted edges~simply~requires to replace the Bernoulli likelihood with, e.g., a Poisson one and consider a latent space construction for the  log-rates, rather than for the logit~of~the~edge~probabilities. 

As discussed in Sections~\ref{sec_phy_22}--\ref{sec_phy_23}, the \textsc{phylnet} model benefits from the availability of multiple~network observations to achieve effective learning of tree-structured network organization.~In~fact,~the~diffuse prior on the tree combined with the complexity of such an object have the effect of limiting~posterior concentration when a single network is observed. One possibility to address~this challenge~is~to supervise the prior on the tree via external node attributes that are often available in practice~and~generally inform on network connectivity. This perspective is equivalent to joint modeling~of~the~node~latent features and the node attributes as processes~over~the~same~tree~$\Upsilon$,~thereby~augmenting~the amount of signal from the data that is necessary to infer the tree.~We~partially explored this direction obtaining promising preliminary results, which encourage further~investigations. In the multiple network setting, another interesting direction would be to explore alternative specifications to the one in Section~\ref{sec_phy_23}. A sensible one consists in replacing~\eqref{eq:bbm_row_multi} with \smash{$ (\bZ^{(m)\intercal}_{[k]}\mid \mu^{(m)}_k, \sigma^2,\Upsilon) \overset{\text{ind}}{\sim}\mbox{N}_V(\bmu_k=[\mu_{k1}, \ldots, \mu_{kV}]^{\intercal},\sigma^2{\bf I}_V)$}, and then assume~that~each~$\bmu_k$, $k=1, \ldots, K$, is from a branching Brownian motion parameterized~by~the tree~$\Upsilon$. By employing shared vectors $\bmu_k$, $k=1, \ldots, K$, this formulation partially reduces flexibility, but introduces additional dependence across the $M$ networks, which might be useful in reducing~posterior~uncertainty on $\Upsilon$. Notice that it is also possible to replace the assumption of equal prior variance~$\sigma^2$ for the latent features with node-specific ones $\sigma_1^2, \ldots, \sigma^2_V$, thereby including further heterogeneity. This extension coincides with removing the ultrametric property for $\Upsilon$ to gain additional~flexibility. While interesting, it is unclear how enlarging the set of tree topologies would affect efficiency of inference on $\Upsilon$. In fact, classical latent space representations \citep{hoff2002latent} rely on a common prior variance for the features, rather than node-specific ones.

Finally, from a more practical perspective, it is of interest to consider further research~on~the~specification of the latent space itself. To this end, particularly relevant is the choice~of~the~dimension $K$, that we address in the article by inheriting common-practice and goodness-of-fit~analyses in classical latent space models \citep{hoff2002latent,handcock2007model,latentnet}.~As~such, any advancement along these lines \citep[e.g.,][]{kaur2023latent}~can~be~readily~employed under the   \textsc{phylnet} model. In addition, inspired by recent research on alternative latent space geometries  \citep[e.g.,][]{smith2019geometry}, we believe it would be of interest to study the possibility of specifying the  \textsc{phylnet}~model under non-Euclidean geometries. In this respect, a relevant direction would be to explore the hyperbolic latent space  \citep{krioukov2010hyperbolic,lubold2023identifying}, as this choice seems to naturally accommodate tree-like structures~in~the~represented~networks.

\vspace{5pt}
\subsubsection*{Acknowledgments}
Federico Pavone is funded by the European Union’s Horizon 2020 research and innovation programme under the Marie Skłodowska-Curie grant agreement No 101034255. Daniele Durante is funded~by~the European Union (ERC, NEMESIS, project number: 101116718). Robin J. Ryder is funded by the European Union under the GA 101071601, through the 2023--2029 ERC Synergy grant OCEAN. Views and opinions expressed are however those of the author(s) only and do not necessarily reflect those~of the European Union or the European Research Council Executive Agency. Neither the European~Union nor the granting authority can be held responsible for them.

\vspace{15pt}
\begin{center}
\Large {\bf Supplementary Material}
\end{center}
\spacingset{1.7} 
\vspace{-35pt}

\setcounter{section}{0}
\setcounter{equation}{0}
\setcounter{table}{0}
\setcounter{figure}{0}

\renewcommand{\theequation}{S.\arabic{equation}}
\renewcommand{\thesection}{S\arabic{section}}  
\renewcommand{\thefigure}{S.\arabic{figure}}
\renewcommand{\thetable}{S.\arabic{table}}

\section{Proofs}\label{sec_1s}
Below we provide the proofs of the theoretical results in the main article. To this end, let us first state and prove two useful lemmas.

\begin{lems}\label{lemma:d}
    Consider prior \eqref{eq:bbm_row_multi} for the features $\bZ^{(m) \intercal}_{[k]}$, $k=1, \ldots, K$ and $m=1, \ldots, M$. Then, conditionally on $\sigma^2$ and $\Upsilon$, it holds
    \begin{eqnarray*}
     \frac{(d^{(m)}_{vu})^2}{2\sigma^2(1-[\bSigma_\Upsilon]_{vu})}=\frac{([\bD]_{vu}^{(m)})^2 }{2\sigma^2(1-[\bSigma_\Upsilon]_{vu})}=\frac{\lVert \bz^{(m)}_{u} - \bz^{(m)}_{u} \rVert^2}{2\sigma^2(1-[\bSigma_\Upsilon]_{vu})}\sim \chi_K^2,
    \end{eqnarray*}
    independently over $m = 1,\dots, M$, where $(d^{(m)}_{vu})^2=([\bD]_{vu}^{(m)})^2=\lVert \bz^{(m)}_{u} - \bz^{(m)}_{u} \rVert^2$ is the squared Euclidean distance between the $K$-dimensional feature vectors of nodes $v$ and $u$ in network $m$, and $ \chi_K^2$~is~a~Chi-square with $K$ degrees of freedom.
\end{lems}

\begin{proof}[Proof of Lemma \ref{lemma:d}]
Recall that under prior  \eqref{eq:bbm_row_multi}, we have
\begin{equation}
    ( \bZ^{(m) \intercal}_{[k]}=[z^{(m)}_{k1}, \ldots, z^{(m)}_{kV}]^{\intercal}  \mid \sigma^2, \Upsilon) \sim\mbox{N}_V(\mu^{(m)}_k\bone_V,\sigma^2\bSigma_\Upsilon)
     \end{equation}
independently for $k=1,\dots,K$ and $m=1,\dots,m$.  Therefore, leveraging standard properties~of~multivariate Gaussian distributions, we have that
\begin{equation}
    (z^{(m)}_{kv} - z^{(m)}_{ku}\mid \sigma^2, \Upsilon) \sim \mbox{N}\left(0, 2\sigma^2(1-[\bSigma_\Upsilon]_{vu})\right),
\end{equation}
and, hence,
\begin{equation}\label{eq:sq_diff}
 ( (z^{(m)}_{kv} - z^{(m)}_{ku})^{2}/(2\sigma^2(1-[\bSigma_\Upsilon]_{vu})) \mid \sigma^2, \Upsilon) \sim \chi_1^2.
\end{equation}
independently for $k=1, \ldots, K$ and $m=1, \ldots, M$. Therefore, since $(d^{(m)}_{vu})^2=\sum_{k=1}^K(z^{(m)}_{kv} - z^{(m)}_{ku})^{2}$, by the properties of the Chi-square  distribution, it holds
\begin{equation}
 \frac{(d^{(m)}_{vu})^2}{2\sigma^2(1-[\bSigma_\Upsilon]_{vu})}=\frac{([\bD]_{vu}^{(m)})^2 }{2\sigma^2(1-[\bSigma_\Upsilon]_{vu})}=\frac{\lVert \bz^{(m)}_{u} - \bz^{(m)}_{u} \rVert^2}{2\sigma^2(1-[\bSigma_\Upsilon]_{vu})}\sim \chi_K^2,
\end{equation}
independently over $m=1,\dots,M$, thereby proving Lemma \ref{lemma:d}.
\end{proof}

\begin{lems}\label{lemma:split_root}
Consider a set of indices $\mathbb{V}=\{1,\dots,V\}$ representing the leaves of two binary splitting trees $\Upsilon_a$ and $\Upsilon_b$. There always exists a pair of indices $(v,u)\in \mathbb{V}\times \mathbb{V}$ such that their most recent common ancestor is the root in both trees $\Upsilon_a$ and $\Upsilon_b$.
\end{lems}

\begin{proof}[Proof of Lemma \ref{lemma:split_root}]
The root bifurcation can be represented by a bipartition of $\mathbb{V}$ in two non-empty complementary sets, which we refer to as $\mathbb{V}=A\cup A^C$ for $\Upsilon_a$, and $\mathbb{V}=B\cup B^C$ for $\Upsilon_b$.~We~want to show that we can always find a pair of indices $(v,u)\in \mathbb{V} \times \mathbb{V}$, such that $(v,u)\in (A\times A^C)\cup (A^C \times A)$ and $(v,u)\in (B\times B^C)\cup (B^C \times B)$. To this end, notice that there are only two possible scenarios representing the relations between the two partitions $(A,A^C)$ and $(B,B^C)$. In one case, the following~holds
\begin{equation}\label{eq:set_cond1}
    A\cap B\neq \emptyset \qquad \text{and} \qquad A^C \cap B^C\neq \emptyset,
\end{equation}
and we can choose $v\in A\cap B$ and $u\in A^C\cap B^C$, which satisfy the required condition.  If  \eqref{eq:set_cond1} does not hold, we can assume without loss of generality that $A^C \cap B^C = \emptyset$ and hence $B^C\subset A$. In this case, it must hold that 
\begin{equation}
A\cap B^C \neq \emptyset \qquad \text{and} \qquad A^C\cap B \neq \emptyset,
\end{equation}
and we can choose $v \in A\cap B^C$ and $u \in A^C\cap B$, which satisfy the required condition.
\end{proof}

Leveraging the above results, let us now prove Lemmas~\ref{prop:constant}--\ref{thm:ident_lat}, along with Theorems~\ref{thm:s2_g}--\ref{thm:consist} and Proposition~\ref{prop:matrix_normal} in the main article.

\vspace{15pt}
\begin{proof}[Proof of Lemma~\ref{prop:constant}]
The proof is constructed considering two cases: $c > 0$ and $ -\delta \leq c < 0$.
\begin{enumerate}
\item \textbf{Case 1. $c>0$}. Consider the centering matrix $\bC=\bI_V - \frac{1}{V}\bJ_V$, where $\bJ_V$ is the $V\times V$ matrix of all ones. Notice that $\bC \bC=\bC$.
The Gram matrix associated to the matrix of squared distances $\bD^2$ can be written as
\begin{equation}
    \bG=-\frac{1}{2}\bC\bD^2\bC.
\end{equation}
Observe that $\widetilde{\bD}^2 = \bD^2 + c^2(\bJ_V-\bI_V) +2c\bD$, where $\bD$ is the element--wise square root of $\bD^2$. Since, $\bC(\bJ_V - \bI_V)\bC = -\bC$, the Gram matrix associated to $\widetilde{\bD}^2$ can be written as
\begin{equation}\label{eq:G_tilde}
\begin{aligned}
    \widetilde{\bG} &=-\frac{1}{2}\bC\bD^2\bC -c\bC\bD\bC + \frac{c^2}{2}\bC = \bG + 2c\bG_{1/2} + \frac{c^2}{2}\bC,
\end{aligned}
\end{equation}
where $\bG_{1/2}$ is the Gram matrix associated to the squared root distance $\bD$. Since $\bD^2$ is Euclidean, then $\bD$ is Euclidean \citep{schoenberg1937certain,maehara2013euclidean} and therefore $\bG_{1/2}$ is positive semidefinite. It follows that $\widetilde{\bG}$ is positive semidefinite as all three terms in the right--hand--side of \eqref{eq:G_tilde} are positive semidefinite. Hence, $\widetilde{\bD}^2$ is a Euclidean distance matrix.

Now, recall that the Schoenberg criterion \citep{schoenberg1935remarks,maehara2013euclidean} states that $\bD^2$~is~representable in $\R^K$ if and only if $\text{rank}(\bG)\leq K$. We now show that if $V$ is large enough, then $\text{rank}(\widetilde{\bG})>K$ and therefore $\widetilde{\bD}^2$ does not admit a representation in $\R^K$.  First of all, notice that by construction the vector $\bone_V$ of all ones is an eigenvector with eigenvalue 0 for both $\widetilde{\bG}$ and $\bG$. Since $\bD^2$ is representable in $\R^K$,~then $\text{rank}(\bG)=r_\bG\leq K$. Thus, there exist $\bv_1,\dots,\bv_{V-r_\bG-1}\ind \bone_V$ vectors such that $\bG\bv_i=\bzero$.
Moreover, $\bC\bv_i = \bv_i$ because of orthogonality to $\bone_V$.
Consider now,
\begin{equation}
\begin{aligned}
    \widetilde{\bG}\bv_i &= \bG\bv_i + 2c\bG_{1/2}\bv_i + \frac{c^2}{2}\bC\bv_i = 2c\bG_{1/2}\bv_i + \frac{c^2}{2}\bv_i.
\end{aligned}
\end{equation}
We show that for all $i=1,\dots,V-r_\bG-1$, the vector $\bv_i$ does not belong to the kernel of $\widetilde{\bG}$. Indeed,
\begin{equation}
    \begin{aligned}
        \bv_i\in\text{Ker}(\widetilde{\bG}) \, &\iff \, 2c\bG_{1/2}\bv_i = - \frac{c^2}{2}\bv_i 
        \iff \, \bG_{1/2}\bv_i = -\frac{c}{4}\bv_i.
    \end{aligned}
\end{equation}
But since $\bG_{1/2}$ is positive semidefinite, it cannot have a negative eigenvalue equal to $-c/4$. Therefore, there are at least $V-r_\bG-1$ orthogonal vectors in $\text{Im}(\widetilde{\bG})$.
Hence, $\text{rank}(\widetilde{\bG})\geq V-r_\bG-1$. Since $r_\bG$ is upper bounded by $K$, for any $V>2K+1$ it holds that $\text{rank}(\widetilde{\bG})>K$.

\item \textbf{Case 2. $-\delta \leq c < 0$}. This case is proven by contradiction. To this end, let us assume~that~$\widetilde{\bD}^2$ admits a representation in $\R^K$. Then, under~\eqref{D_id}, $\bD$ can be written as $[\bD]_{vu}=[\widetilde{\bD}]_{vu} + (-c)$~with $ -c > 0$. Leveraging this result and  the arguments of Case 1, it follows that, if~{$V>2K{+}1$}, $\bD^2$ does not admit a representation in $\R^K$. This is a contradiction, because $\bD^2$ admits a representation in $\R^K$ by construction. Hence, $\widetilde{\bD}^2$ does not admit a representation in $\R^K$ if $V > 2K + 1$.
\end{enumerate}
\vspace{-20pt}
\end{proof}

\vspace{-10pt}
\begin{proof}[Proof of Lemma~\ref{thm:ident_lat}]
Under~\eqref{eq:lik_single}, the joint distribution of the edges in $\bY$ is identified by the Bernoulli probabilities $\theta_{vu}$, $v=2, \ldots, V$, $u=1, \ldots, v-1$, which in turn depend on $a$ and $\bD$, under the 1-to-1 logit mapping, via the predictor $a-\lVert \bz_{u} - \bz_{u} \rVert=a-[\bD]_{vu}$.
 Therefore, the assumption 
\begin{equation}
    \Prob_{a, \bD} \overset{\text{d}}{=} \Prob_{\widetilde{a},\widetilde{\bD}}
\end{equation}
implies that $a-([\bD]_{vu})=\widetilde{a}-[\widetilde{\bD}]_{vu}$ for every $v=2, \ldots, V$, $u=1, \ldots, v-1$. 
Without loss of generality, we can assume $ \widetilde{a} = a + c$, for some constant $c\neq0$. 
Hence, in order for the equality $a-[\bD]_{vu}=\widetilde{a}-[\widetilde{\bD}]_{vu}$ to hold, we need that $a-[\bD]_{vu}= a+c-[\widetilde{\bD}]_{vu}$ which implies $[\widetilde{\bD}]_{vu}=[\bD]_{vu}+c$. 
If $c < -\delta$, with $\delta = \min \{[\bD]_{vu}: v=2,\dots,V, \, u=1,\dots,v-1\}$, then $\widetilde{\bD}$ would have at least one negative entry and it would not be a pairwise distance matrix. Otherwise, if $c\geq -\delta$, by Lemma~\ref{prop:constant}, this is not possible as $\widetilde{\bD}^2$ would not be representable in $\R^K$.
\end{proof}

\vspace{10pt}
\begin{proof}[Proof of Theorem \ref{thm:s2_g}]
We prove the result by showing that if $(\sigma^2,\Upsilon)\neq (\widetilde{\sigma}^2,\widetilde{\Upsilon})$, then \begin{equation}\label{eq:ineq}
\Prob_{\sigma^2,\Upsilon} \overset{\text{d}}{\neq} \Prob_{\widetilde{\sigma}^2,\widetilde{\Upsilon}}.
\end{equation}
Leveraging Lemma \ref{lemma:d} and defining $\omega_{vu}(\sigma^2,\Upsilon):=2\sigma^2(1-[\bSigma_\Upsilon]_{vu})$, we can write the probability of a success in the marginal model as
\begin{equation}
   \mbox{pr}(y_{vu}^{(m)} = 1\mid \sigma^2,\Upsilon) = \E_{a}[ \E_{(\bar{d}^{(m)}_{vu})^2}[ \text{expit}\{a - \sqrt{(\bar{d}^{(m)}_{vu})^2 \omega_{vu}(\sigma^2,\Upsilon)}\}\mid a ]],
\end{equation}
where $\text{expit}(x)=\exp(x)/[1+\exp(x)]$, while the expectation is computed with respect to $(\bar{d}^{(m)}_{vu})^2 \sim\chi^2_K$ and $a\sim \mbox{N}(0, \sigma_a^2)$, i.e., the prior distributions. In order to prove \eqref{eq:ineq}, we show that for at least one pair $(v,u)$ the probability of an edge between $v$ and $u$ is different under the two parametrizations.~In particular, when $(\sigma^2,\Upsilon)\neq(\widetilde{\sigma}^2,\widetilde{\Upsilon})$ we can always choose $(v,u)$ such that $\omega_{vu}(\sigma^2,\Upsilon)\neq \omega_{vu}(\widetilde{\sigma}^2,\widetilde{\Upsilon})$ as follows:
\begin{itemize}
   \item If $\sigma^2\neq \widetilde{\sigma}^2$, then from Lemma \ref{lemma:split_root} there exists $(v,u)$ such that the two nodes split at the root in both $\Upsilon$ and $\widetilde{\Upsilon}$. Therefore $[\bSigma_\Upsilon]_{vu}=[\bSigma_{\widetilde{\Upsilon}}]_{vu}=0$, and, hence
   \begin{equation}
       \omega_{vu}(\sigma^2,\Upsilon) = 2\sigma^2 \neq 2\widetilde{\sigma}^2 = \omega_{vu}(\widetilde{\sigma}^2,\widetilde{\Upsilon});
   \end{equation}
   \item If $\sigma^2 = \widetilde{\sigma}^2$ and $\Upsilon\neq \widetilde{\Upsilon}$, then there exists $(v,u)$ such that $[\bSigma_\Upsilon]_{vu} \neq [\bSigma_{\widetilde{\Upsilon}}]_{vu}$. Therefore, 
   \begin{equation}
       \omega_{vu}(\sigma^2,\Upsilon)=2\sigma^2(1-[\bSigma_\Upsilon]_{vu}) \neq 2\sigma^2(1-[\bSigma_{\widetilde{\Upsilon}}]_{vu}) = \omega_{vu}(\widetilde{\sigma}^2,\widetilde{\Upsilon}).
   \end{equation}
\end{itemize}
\vspace{-5pt}
Without loss of generality, assume that $\omega_{vu}(\sigma^2,\Upsilon)> \omega_{vu}(\widetilde{\sigma}^2,\widetilde{\Upsilon})$. This implies $$ \text{expit}\{a - \sqrt{(\bar{d}^{(m)}_{vu})^2 \omega_{vu}(\sigma^2,\Upsilon)}\}-\text{expit}\{a - \sqrt{(\bar{d}^{(m)}_{vu})^2 \omega_{vu}(\widetilde{\sigma}^2,\widetilde{\Upsilon})}\}<0,$$
for any $(\bar{d}^{(m)}_{vu})^2>0$. Therefore, for the pair of nodes $(v,u)$, it follows
\begin{eqnarray*}
\begin{split}
& \mbox{pr}(y_{vu}^{(m)} = 1\mid \sigma^2,\Upsilon) - \mbox{pr}(y_{vu}^{(m)} = 1\mid \widetilde{\sigma}^2,\widetilde{\Upsilon}) \\
 &=\E_{a}[ \E_{(\bar{d}^{(m)}_{vu})^2}[ \text{expit}\{a - \sqrt{(\bar{d}^{(m)}_{vu})^2 \omega_{vu}(\sigma^2,\Upsilon)}\}-\text{expit}\{a - \sqrt{(\bar{d}^{(m)}_{vu})^2 \omega_{vu}(\widetilde{\sigma}^2,\widetilde{\Upsilon})}\} \mid a ]] < 0,
 \end{split}
\end{eqnarray*}
for any $a \in \R$, thereby proving Theorem \ref{thm:s2_g}.
\end{proof}

\vspace{15pt}

\begin{proof}[Proof of Proposition~\ref{prop:matrix_normal}]
First notice that, under \eqref{eq:bbm_row}, $((\bZ -\bmu \otimes \bone^{\intercal}_V) \mid \sigma^2, \Upsilon) $ is distributed as a matrix-normal. Specifically, $((\bZ -\bmu \otimes \bone^{\intercal}_V) \mid \sigma^2, \Upsilon) \sim \mbox{MN}_{K \times V}(\bzero_{K\times V}, \bI_K, \sigma^2\bSigma_\Upsilon) $. Therefore, leveraging~standard properties of matrix-normal distributions, we have that 
\begin{eqnarray*}
    (\bR(\bZ -\bmu \otimes \bone^{\intercal}_V) \mid \sigma^2, \Upsilon) \sim \mbox{MN}_{K \times V}(\bzero_{K\times V}, \bR\bI_K \bR^{\intercal}, \sigma^2\bSigma_\Upsilon).
\end{eqnarray*}
Since $\bR$ is an orthogonal matrix, it holds $\bR\bI_K \bR^{\intercal}=\bI_K$, which implies  $( \bR(\bZ-\bmu \otimes {\bf 1}^\intercal_{V} )\mid \sigma^2,\Upsilon) \overset{\text{d}}{=} (\bZ-\bmu \otimes {\bf 1}^\intercal_{V} \mid \sigma^2,\Upsilon)$.
\end{proof}

\vspace{15pt}

\begin{proof}[Proof of Theorem \ref{thm:consist}]
Theorem \ref{thm:consist} is a direct consequence of Doob's consistency \citep{doob1949application} and the finiteness of the space $\mathbb{T}_V$ of binary tree topologies. Following the identifiability of $(\sigma^2,\Upsilon)$~for~the~marginal model $\mathbb{P}_{\sigma^2,\Upsilon}$ proved in Theorem~\ref{thm:s2_g}, Doob's consistency theorem guarantees posterior consistency almost surely with respect to the prior $\Pi_{\sigma^2,\Upsilon}$, where $\Upsilon=(\bm{\lambda},\mathcal{T})$. 

    Denote with $\mathbb{S} \subset\mathbb{R}^+\times{\mathbb{R}^{+}}^{2V-2}\times\mathbb{T}_V$ the subspace of parameters $(\sigma^2,\bm{\lambda},\mathcal{T})$ for which posterior consistency holds everywhere due to Doob's, i.e. $\Pi_{\sigma^2,\bm{\lambda},\mathcal{T}}(\mathbb{S})=1$ where $\Pi_{\sigma^2,\bm{\lambda},\mathcal{T}}$ denotes the joint prior on $(\sigma^2,\bm{\lambda},\mathcal{T})$ defined in Section~\ref{sec_phy_23} of the main article. It follows that 
    \begin{equation}\label{eq:cons_1}
        1=\Pi_{\sigma^2,\bm{\lambda},\mathcal{T}}(\mathbb{S} )=\E_{\sigma^2,\bm{\lambda},\mathcal{T}}[ \mathcal{I}_{\mathbb{S}}(\sigma^2,\bm{\lambda},\mathcal{T})]=\sum\nolimits_{j=1}^{\mid\mathbb{T}_V\mid}\E_{\sigma^2,\bm{\lambda}\mid \mathcal{T}=\mathcal{T}_j}[ \mathcal{I}_{\mathbb{S}}(\sigma^2,\bm{\lambda},\mathcal{T})]\Pi_\mathcal{T}(\{\mathcal{T}_j\}),
    \end{equation}
    where $\mathcal{I}$ is the indicator function. Since $\Pi_\mathcal{T}(\{\mathcal{T}_j\})>0$, $\forall\mathcal{T}_j\in\mathbb{T}_V$, due to the full support of the prior we consider, \eqref{eq:cons_1} implies that
    \begin{equation}\label{eq:cons_2}
         \E_{\sigma^2,\bm{\lambda}\mid\mathcal{T}=\mathcal{T}_j}[ \mathcal{I}_{\mathbb{S}}(\sigma^2,\bm{\lambda},\mathcal{T})]=1, \quad \forall \mathcal{T}_j\in\mathbb{T}_V.
    \end{equation}
    Therefore, $\forall\mathcal{T}_j\in\mathbb{T}_V$ it follows that $ \mathcal{I}_{\mathbb{S} }(\sigma^2,\bm{\lambda},\mathcal{T}_j)=1$ $a.s.-\Pi_{\sigma^2,\bm{\lambda}\mid\mathcal{T}_j}$. Finally, $\forall\mathcal{T}_0\in\mathbb{T}_V$ one can choose $H_0=\{\sigma^2,\bm{\lambda}: \mathcal{I}_{\mathbb{S} }(\sigma^2,\lambda,\mathcal{T}_0)=1\}$, for which it holds that $\Pi_{\sigma^2,\bm{\lambda}\mid\mathcal{T}_0}(H_0)=1$ and that $H_0\times\{\mathcal{T}_0\}\subset\mathbb{S} $.
\end{proof}

\vspace{10pt}
\let\oldbibliography\thebibliography
\renewcommand{\thebibliography}[1]{\oldbibliography{#1}
	\setlength{\itemsep}{4.5pt}} 

\spacingset{1.5}
\begingroup
\fontsize{11pt}{11pt}\selectfont
\bibliographystyle{JASA}
\bibliography{phynet.bib}

\begin{thebibliography}{82}
\newcommand{\enquote}[1]{``#1''}
\expandafter\ifx\csname natexlab\endcsname\relax\def\natexlab#1{#1}\fi

\bibitem[{Airoldi et~al.(2008)Airoldi, Blei, Fienberg, and
  Xing}]{airoldi2008mixed}
Airoldi, E.~M., Blei, D., Fienberg, S.,  and Xing, E. (2008), \enquote{Mixed
  membership stochastic blockmodels,} \textit{Journal of Machine Learning
  Research}, 9, 1981--2014.

\bibitem[{Aldous(2001)}]{aldous2001stochastic}
Aldous, D.~J. (2001), \enquote{{Stochastic models and descriptive statistics
  for phylogenetic trees, from Yule to today},} \textit{Statistical Science},
  16, 23--34.

\bibitem[{Andrieu and Thoms(2008)}]{andrieu2008tutorial}
Andrieu, C.,  and Thoms, J. (2008), \enquote{{A tutorial on adaptive MCMC},}
  \textit{Statistics and Computing}, 18, 343--373.

\bibitem[{Arroyo et~al.(2021)Arroyo, Athreya, Cape, Chen, Priebe, and
  Vogelstein}]{arroyo2021inference}
Arroyo, J., Athreya, A., Cape, J., Chen, G., Priebe, C.~E.,  and Vogelstein,
  J.~T. (2021), \enquote{Inference for multiple heterogeneous networks with a
  common invariant subspace,} \textit{Journal of Machine Learning Research},
  22, 6303--6351.

\bibitem[{Athreya et~al.(2018)Athreya, Fishkind, Tang, Priebe, Park,
  Vogelstein, Levin, Lyzinski, Qin, and Sussman}]{athreya2018statistical}
Athreya, A., Fishkind, D.~E., Tang, M., Priebe, C.~E., Park, Y., Vogelstein,
  J.~T., Levin, K., Lyzinski, V., Qin, Y.,  and Sussman, D.~L. (2018),
  \enquote{Statistical inference on random dot product graphs: a survey,}
  \textit{Journal of Machine Learning Research}, 18, 1--92.

\bibitem[{Betzel and Bassett(2017)}]{betzel2017multi}
Betzel, R.~F.,  and Bassett, D.~S. (2017), \enquote{Multi-scale brain
  networks,} \textit{Neuroimage}, 160, 73--83.

\bibitem[{Blaizot et~al.(2010)Blaizot, Mansilla, Insausti, Constans,
  Salinas-Alaman, Pro-Sistiaga, Mohedano-Moriano, and
  Insausti}]{blaizot2010human}
Blaizot, X., Mansilla, F., Insausti, A., Constans, J., Salinas-Alaman, A.,
  Pro-Sistiaga, P., Mohedano-Moriano, A.,  and Insausti, R. (2010),
  \enquote{{The human parahippocampal region: I. Temporal pole
  cytoarchitectonic and MRI correlation},} \textit{Cerebral Cortex}, 20,
  2198--2212.

\bibitem[{B{\"o}cker et~al.(2013)B{\"o}cker, Canzar, and
  Klau}]{bocker2013generalized}
B{\"o}cker, S., Canzar, S.,  and Klau, G.~W. (2013), \enquote{{The generalized
  Robinson-Foulds metric},} in \textit{Algorithms in Bioinformatics: 13th
  International Workshop, WABI 2013, Sophia Antipolis, France, September 2-4,
  2013. Proceedings 13}, Springer, pp. 156--169.

\bibitem[{Borgatti et~al.(1990)Borgatti, Everett, and Shirey}]{borgatti1990ls}
Borgatti, S.~P., Everett, M.~G.,  and Shirey, P.~R. (1990), \enquote{{LS} sets,
  lambda sets and other cohesive subsets,} \textit{Social Networks}, 12,
  337--357.

\bibitem[{Borgs et~al.(2018)Borgs, Chayes, Cohn, and Holden}]{borgs2018sparse}
Borgs, C., Chayes, J.~T., Cohn, H.,  and Holden, N. (2018), \enquote{Sparse
  exchangeable graphs and their limits via graphon processes,} \textit{Journal
  of Machine Learning Research}, 18, 1--71.

\bibitem[{Bouckaert(2010)}]{bouckaert2010densitree}
Bouckaert, R.~R. (2010), \enquote{{DensiTree: Making sense of sets of
  phylogenetic trees},} \textit{Bioinformatics}, 26, 1372--1373.

\bibitem[{Bright et~al.(2022)Bright, Brewer, and Morselli}]{bright2022reprint}
Bright, D., Brewer, R.,  and Morselli, C. (2022), \enquote{Using social network
  analysis to study crime: Navigating the challenges of criminal justice
  records,} \textit{Social Networks}, 69, 235--250.

\bibitem[{Buckley et~al.(2025)Buckley, Kopp, Pellard, Ryder, and
  Jacques}]{buckley2025contrasting}
Buckley, C.~D., Kopp, E., Pellard, T., Ryder, R.~J.,  and Jacques, G. (2025),
  \enquote{Contrasting modes of cultural evolution: Kra-Dai languages and
  weaving technologies,} \textit{Evolutionary Human Sciences}, 7, 1--28.

\bibitem[{Bullmore and Sporns(2009)}]{bullmore2009complex}
Bullmore, E.,  and Sporns, O. (2009), \enquote{{Complex brain networks: Graph
  theoretical analysis of structural and functional systems},} \textit{Nature
  Reviews Neuroscience}, 10, 186--198.

\bibitem[{Bullmore and Sporns(2012)}]{bullmore2012economy}
Bullmore, E.--- (2012), \enquote{The economy of brain network organization,}
  \textit{Nature Reviews Neuroscience}, 13, 336--349.

\bibitem[{Calderoni et~al.(2017)Calderoni, Brunetto, and
  Piccardi}]{calderoni2017communities}
Calderoni, F., Brunetto, D.,  and Piccardi, C. (2017), \enquote{Communities in
  criminal networks: A case study,} \textit{Social Networks}, 48, 116--125.

\bibitem[{Campana and Varese(2022)}]{campana2022studying}
Campana, P.,  and Varese, F. (2022), \enquote{Studying organized crime
  networks: Data sources, boundaries and the limits of structural measures,}
  \textit{Social Networks}, 69, 149--159.

\bibitem[{Caron and Fox(2017)}]{caron2017sparse}
Caron, F.,  and Fox, E.~B. (2017), \enquote{Sparse graphs using exchangeable
  random measures,} \textit{Journal of the Royal Statistical Society Series B:
  Statistical Methodology}, 79, 1295--1366.

\bibitem[{Catino(2014)}]{catino2014mafias}
Catino, M. (2014), \enquote{How Do Mafias Organize?: Conflict and Violence in
  Three Mafia Organizations,} \textit{European Journal of Sociology/Archives
  Europ{\'e}ennes de Sociologie}, 55, 177--220.

\bibitem[{Chen et~al.(2014)Chen, Kuo, and Lewis}]{chen2014bayesian}
Chen, M.-H., Kuo, L.,  and Lewis, P.~O. (2014), \textit{{Bayesian
  Phylogenetics: Methods, Algorithms, and Applications}}, CRC Press.

\bibitem[{Clauset et~al.(2008)Clauset, Moore, and
  Newman}]{clauset2008hierarchical}
Clauset, A., Moore, C.,  and Newman, M.~E. (2008), \enquote{Hierarchical
  structure and the prediction of missing links in networks,} \textit{Nature},
  453, 98--101.

\bibitem[{Coutinho et~al.(2020)Coutinho, Divi{\'a}k, Bright, and
  Koskinen}]{coutinho2020multilevel}
Coutinho, J.~A., Divi{\'a}k, T., Bright, D.,  and Koskinen, J. (2020),
  \enquote{Multilevel determinants of collaboration between organised criminal
  groups,} \textit{Social Networks}, 63, 56--69.

\bibitem[{Craddock et~al.(2013)Craddock, Jbabdi, Yan, Vogelstein, Castellanos,
  Di~Martino, Kelly, Heberlein, Colcombe, and Milham}]{craddock2013imaging}
Craddock, R.~C., Jbabdi, S., Yan, C.-G., Vogelstein, J.~T., Castellanos, F.~X.,
  Di~Martino, A., Kelly, C., Heberlein, K., Colcombe, S.,  and Milham, M.~P.
  (2013), \enquote{Imaging human connectomes at the macroscale,} \textit{Nature
  Methods}, 10, 524--539.

\bibitem[{Desikan et~al.(2006)Desikan, S{\'e}gonne, Fischl, Quinn, Dickerson,
  Blacker, Buckner, Dale, Maguire, Hyman, et~al.}]{desikan2006automated}
Desikan, R.~S., S{\'e}gonne, F., Fischl, B., Quinn, B.~T., Dickerson, B.~C.,
  Blacker, D., Buckner, R.~L., Dale, A.~M., Maguire, R.~P., Hyman, B.~T. et~al.
  (2006), \enquote{{An automated labeling system for subdividing the human
  cerebral cortex on MRI scans into gyral based regions of interest},}
  \textit{NeuroImage}, 31, 968--980.

\bibitem[{Divi{\'a}k(2022)}]{diviak2022key}
Divi{\'a}k, T. (2022), \enquote{Key aspects of covert networks data collection:
  Problems, challenges, and opportunities,} \textit{Social Networks}, 69,
  160--169.

\bibitem[{Doob(1949)}]{doob1949application}
Doob, J.~L. (1949), \enquote{Application of the theory of martingales,}
  \textit{{Le Calcul Des Probabilit\'es et Ses Applications}}, 23--27.

\bibitem[{Durante and Dunson(2014)}]{durante2014nonparametric}
Durante, D.,  and Dunson, D.~B. (2014), \enquote{{Nonparametric Bayes dynamic
  modelling of relational data},} \textit{Biometrika}, 101, 883--898.

\bibitem[{Durante et~al.(2017)Durante, Dunson, and
  Vogelstein}]{durante2017nonparametric}
Durante, D., Dunson, D.~B.,  and Vogelstein, J.~T. (2017),
  \enquote{{Nonparametric Bayes modeling of populations of networks},}
  \textit{Journal of the American Statistical Association}, 112, 1516--1530.

\bibitem[{Eastman et~al.(2011)Eastman, Alfaro, Joyce, Hipp, and
  Harmon}]{eastman2011novel}
Eastman, J.~M., Alfaro, M.~E., Joyce, P., Hipp, A.~L.,  and Harmon, L.~J.
  (2011), \enquote{A novel comparative method for identifying shifts in the
  rate of character evolution on trees,} \textit{Evolution}, 65, 3578--3589.

\bibitem[{Esfahlani et~al.(2021)Esfahlani, Jo, Puxeddu, Merritt, Tanner,
  Greenwell, Patel, Faskowitz, and Betzel}]{esfahlani2021modularity}
Esfahlani, F.~Z., Jo, Y., Puxeddu, M.~G., Merritt, H., Tanner, J.~C.,
  Greenwell, S., Patel, R., Faskowitz, J.,  and Betzel, R.~F. (2021),
  \enquote{Modularity maximization as a flexible and generic framework for
  brain network exploratory analysis,} \textit{Neuroimage}, 244, 118607.

\bibitem[{Evans et~al.(2021)Evans, Greenhill, Watts, List, Botero, Gray, and
  Kirby}]{evans2021uses}
Evans, C.~L., Greenhill, S.~J., Watts, J., List, J.-M., Botero, C.~A., Gray,
  R.~D.,  and Kirby, K.~R. (2021), \enquote{The uses and abuses of tree
  thinking in cultural evolution,} \textit{Philosophical Transactions of the
  Royal Society B}, 376, 20200056.

\bibitem[{Felsenstein(1981)}]{felsenstein1981evolutionary}
Felsenstein, J. (1981), \enquote{{Evolutionary trees from DNA sequences: A
  maximum likelihood approach},} \textit{Journal of Molecular Evolution}, 17,
  368--376.

\bibitem[{Felsenstein(1985)}]{felsenstein1985phylogenies}
--- (1985), \enquote{Phylogenies and the comparative method,} \textit{The
  American Naturalist}, 125, 1--15.

\bibitem[{Felsenstein(2004)}]{felsenstein2004inferring}
--- (2004), \textit{{Inferring Phylogenies}}, vol.~2, Sinauer Associates
  Sunderland, MA.

\bibitem[{Fosdick et~al.(2019)Fosdick, McCormick, Murphy, Ng, and
  Westling}]{fosdick2019multiresolution}
Fosdick, B.~K., McCormick, T.~H., Murphy, T.~B., Ng, T. L.~J.,  and Westling,
  T. (2019), \enquote{Multiresolution network models,} \textit{Journal of
  Computational and Graphical Statistics}, 28, 185--196.

\bibitem[{Freeman(1992)}]{freeman1992sociological}
Freeman, L.~C. (1992), \enquote{The sociological concept of "group": An
  empirical test of two models,} \textit{American Journal of Sociology}, 98,
  152--166.

\bibitem[{Fruchterman and Reingold(1991)}]{fruchterman1991graph}
Fruchterman, T.~M.,  and Reingold, E.~M. (1991), \enquote{Graph drawing by
  force-directed placement,} \textit{Software: Practice and Experience}, 21,
  1129--1164.

\bibitem[{Ghosal and Van~der Vaart(2017)}]{ghosal2017fundamentals}
Ghosal, S.,  and Van~der Vaart, A. (2017), \textit{Fundamentals of
  Nonparametric Bayesian Inference}, vol.~44, Cambridge University Press.

\bibitem[{Gollini and Murphy(2016)}]{gollini2016joint}
Gollini, I.,  and Murphy, T.~B. (2016), \enquote{Joint modeling of multiple
  network views,} \textit{Journal of Computational and Graphical Statistics},
  25, 246--265.

\bibitem[{Gosztolai and Arnaudon(2021)}]{gosztolai2021unfolding}
Gosztolai, A.,  and Arnaudon, A. (2021), \enquote{{Unfolding the multiscale
  structure of networks with dynamical Ollivier-Ricci curvature},}
  \textit{Nature Communications}, 12, 4561.

\bibitem[{Handcock et~al.(2007)Handcock, Raftery, and
  Tantrum}]{handcock2007model}
Handcock, M.~S., Raftery, A.~E.,  and Tantrum, J.~M. (2007),
  \enquote{Model-based clustering for social networks,} \textit{Journal of the
  Royal Statistical Society Series A: Statistics in Society}, 170, 301--354.

\bibitem[{Harris(1963)}]{harris1963theory}
Harris, T.~E. (1963), \textit{The Theory of Branching Processes}, vol.~6,
  Springer Berlin.

\bibitem[{Hayden et~al.(1991)Hayden, Wells, Liu, and Tarazaga}]{hayden1991cone}
Hayden, T.~L., Wells, J., Liu, W.-M.,  and Tarazaga, P. (1991), \enquote{The
  cone of distance matrices,} \textit{Linear Algebra and its Applications},
  144, 153--169.

\bibitem[{Herlau et~al.(2012)Herlau, M{\o}rup, Schmidt, and
  Hansen}]{herlau2012detecting}
Herlau, T., M{\o}rup, M., Schmidt, M.~N.,  and Hansen, L.~K. (2012),
  \enquote{Detecting hierarchical structure in networks,} \textit{Proceedings
  of the 3rd International Workshop on Cognitive Information Processing (CIP)},
  1--6.

\bibitem[{Hillis et~al.(1994)Hillis, Huelsenbeck, and
  Swofford}]{hillis1994hobgoblin}
Hillis, D.~M., Huelsenbeck, J.~P.,  and Swofford, D.~L. (1994),
  \enquote{Hobgoblin of phylogenetics?} \textit{Nature}, 369, 363--364.

\bibitem[{Hoff et~al.(2002)Hoff, Raftery, and Handcock}]{hoff2002latent}
Hoff, P.~D., Raftery, A.~E.,  and Handcock, M.~S. (2002), \enquote{Latent space
  approaches to social network analysis,} \textit{Journal of the American
  Statistical Association}, 97, 1090--1098.

\bibitem[{Hoffmann et~al.(2021)Hoffmann, Bouckaert, Greenhill, and
  K{\"u}hnert}]{hoffmann2021bayesian}
Hoffmann, K., Bouckaert, R., Greenhill, S.~J.,  and K{\"u}hnert, D. (2021),
  \enquote{Bayesian phylogenetic analysis of linguistic data using BEAST,}
  \textit{Journal of Language Evolution}, 6, 119--135.

\bibitem[{Kaur et~al.(2024)Kaur, Rastelli, Friel, and Raftery}]{kaur2023latent}
Kaur, H., Rastelli, R., Friel, N.,  and Raftery, A.~E. (2024), \enquote{Latent
  Position Network Models,} \textit{The Sage Handbook of Social Network
  Analysis (Second Edition)}, 36, 526--541.

\bibitem[{Kelly et~al.(2023)Kelly, Ryder, and Clart{\'e}}]{kelly2023lagged}
Kelly, L.~J., Ryder, R.~J.,  and Clart{\'e}, G. (2023), \enquote{{Lagged
  couplings diagnose Markov chain Monte Carlo phylogenetic inference},}
  \textit{The Annals of Applied Statistics}, 17, 1419--1443.

\bibitem[{Kiar et~al.(2017)Kiar, Bridgeford, Roncai, for Reliability, (CoRR),
  Chandrashekhar, Mhembere, Ryman, Zuo, Margulies, Craddock,
  et~al.}]{kiar2017high}
Kiar, G., Bridgeford, E.~W., Roncai, W. R.~G., for Reliability, C., (CoRR), R.,
  Chandrashekhar, V., Mhembere, D., Ryman, S., Zuo, X.-N., Margulies, D.~S.,
  Craddock, R.~C. et~al. (2017), \enquote{A high-throughput pipeline identifies
  robust connectomes but troublesome variability,} \textit{Biorxiv}, 188706.

\bibitem[{Kivel{\"a} et~al.(2014)Kivel{\"a}, Arenas, Barthelemy, Gleeson,
  Moreno, and Porter}]{kivela2014multilayer}
Kivel{\"a}, M., Arenas, A., Barthelemy, M., Gleeson, J.~P., Moreno, Y.,  and
  Porter, M.~A. (2014), \enquote{Multilayer networks,} \textit{Journal of
  Complex Networks}, 2, 203--271.

\bibitem[{Krioukov et~al.(2010)Krioukov, Papadopoulos, Kitsak, Vahdat, and
  Bogu\~n{\'a}}]{krioukov2010hyperbolic}
Krioukov, D., Papadopoulos, F., Kitsak, M., Vahdat, A.,  and Bogu\~n{\'a}, M.
  (2010), \enquote{Hyperbolic geometry of complex networks,} \textit{Physical
  Review E}, 82, 036106.

\bibitem[{Krivitsky and Handcock(2008)}]{latentnet}
Krivitsky, P.~N.,  and Handcock, M.~S. (2008), \enquote{Fitting position latent
  cluster models for social networks with latentnet,} \textit{Journal of
  Statistical Software}, 24, 1--13.

\bibitem[{Krivitsky et~al.(2009)Krivitsky, Handcock, Raftery, and
  Hoff}]{krivitsky2009representing}
Krivitsky, P.~N., Handcock, M.~S., Raftery, A.~E.,  and Hoff, P.~D. (2009),
  \enquote{Representing degree distributions, clustering, and homophily in
  social networks with latent cluster random effects models,} \textit{Social
  Networks}, 31, 204--213.

\bibitem[{Kuhner and Felsenstein(1994)}]{kuhner1994}
Kuhner, M.~K.,  and Felsenstein, J. (1994), \enquote{A simulation comparison of
  phylogeny algorithms under equal and unequal evolutionary rates.}
  \textit{Molecular Biology and Evolution}, 11, 459--468.

\bibitem[{Legramanti et~al.(2022)Legramanti, Rigon, Durante, and
  Dunson}]{legramanti2022extended}
Legramanti, S., Rigon, T., Durante, D.,  and Dunson, D.~B. (2022),
  \enquote{Extended stochastic block models with application to criminal
  networks,} \textit{The Annals of Applied Statistics}, 16, 2369--2395.

\bibitem[{Lu et~al.(2025)Lu, Durante, and Friel}]{lu2024zero}
Lu, C., Durante, D.,  and Friel, N. (2025), \enquote{Zero-inflated stochastic
  block modeling of efficiency-security tradeoffs in weighted criminal
  networks,} \textit{Journal of the Royal Statistical Society Series A:
  Statistics in Society}, In Press.

\bibitem[{Lubold et~al.(2023)Lubold, Chandrasekhar, and
  McCormick}]{lubold2023identifying}
Lubold, S., Chandrasekhar, A.~G.,  and McCormick, T.~H. (2023),
  \enquote{Identifying the latent space geometry of network models through
  analysis of curvature,} \textit{Journal of the Royal Statistical Society
  Series B: Statistical Methodology}, 85, 240--292.

\bibitem[{MacDonald et~al.(2022)MacDonald, Levina, and
  Zhu}]{macdonald2022latent}
MacDonald, P.~W., Levina, E.,  and Zhu, J. (2022), \enquote{Latent space models
  for multiplex networks with shared structure,} \textit{Biometrika}, 109,
  683--706.

\bibitem[{Maehara(2013)}]{maehara2013euclidean}
Maehara, H. (2013), \enquote{Euclidean embeddings of finite metric spaces,}
  \textit{Discrete Mathematics}, 313, 2848--2856.

\bibitem[{May and Moore(2020)}]{may2020bayesian}
May, M.~R.,  and Moore, B.~R. (2020), \enquote{A Bayesian approach for
  inferring the impact of a discrete character on rates of continuous-character
  evolution in the presence of background-rate variation,} \textit{Systematic
  Biology}, 69, 530--544.

\bibitem[{Meunier et~al.(2010)Meunier, Lambiotte, and
  Bullmore}]{meunier2010modular}
Meunier, D., Lambiotte, R.,  and Bullmore, E.~T. (2010), \enquote{Modular and
  hierarchically modular organization of brain networks,} \textit{Frontiers in
  Neuroscience}, 4, 200.

\bibitem[{Nowicki and Snijders(2001)}]{nowicki2001estimation}
Nowicki, K.,  and Snijders, T. A.~B. (2001), \enquote{Estimation and prediction
  for stochastic blockstructures,} \textit{Journal of the American Statistical
  Association}, 96, 1077--1087.

\bibitem[{Paoli(2007)}]{paoli2007mafia}
Paoli, L. (2007), \enquote{Mafia and organised crime in Italy: the
  unacknowledged successes of law enforcement,} \textit{West European
  Politics}, 30, 854--880.

\bibitem[{Pu et~al.(2026)Pu, Fan, and Fang}]{pu2025tree}
Pu, D., Fan, X.,  and Fang, K. (2026), \enquote{Tree-Enhanced Latent Space
  Models for Two-Mode Networks,} \textit{Journal of Computational and Graphical
  Statistics}, In press, 1--17.

\bibitem[{Ravasz and Barab{\'a}si(2003)}]{ravasz2003hierarchical}
Ravasz, E.,  and Barab{\'a}si, A.-L. (2003), \enquote{Hierarchical organization
  in complex networks,} \textit{Physical Review E}, 67, 026112.

\bibitem[{Ross(2014)}]{ross2014introduction}
Ross, S.~M. (2014), \textit{Introduction to Probability Models}, Academic
  Press.

\bibitem[{Roy et~al.(2006)Roy, Kemp, Mansinghka, and
  Tenenbaum}]{roy2006learning}
Roy, D.~M., Kemp, C., Mansinghka, V.,  and Tenenbaum, J. (2006),
  \enquote{Learning annotated hierarchies from relational data,}
  \textit{Advances in Neural Information Processing Systems (NeurIPS)}, 19.

\bibitem[{Roy and Teh(2008)}]{roy2008mondrian}
Roy, D.~M.,  and Teh, Y. (2008), \enquote{{The Mondrian process},}
  \textit{Advances in Neural Information Processing Systems (NeurIPS)}, 21,
  1--8.

\bibitem[{Rubinov and Sporns(2010)}]{rubinov2010complex}
Rubinov, M.,  and Sporns, O. (2010), \enquote{Complex network measures of brain
  connectivity: uses and interpretations,} \textit{Neuroimage}, 52, 1059--1069.

\bibitem[{Ryder(2025)}]{phylolinguisticsworkflow}
Ryder, R.~J. (2025), \enquote{Bayesian Methods in Historical Linguistics: a
  Workflow,} \textit{Proceedings of the COMPAS 2024 conference}, 1--18.

\bibitem[{Salter-Townshend and McCormick(2017)}]{salter2017latent}
Salter-Townshend, M.,  and McCormick, T.~H. (2017), \enquote{Latent space
  models for multiview network data,} \textit{The Annals of Applied
  Statistics}, 11, 1217--1244.

\bibitem[{Sarkar and Moore(2006)}]{sarkar2006dynamic}
Sarkar, P.,  and Moore, A.~W. (2006), \enquote{Dynamic social network analysis
  using latent space models,} \textit{Advances in Neural Information Processing
  Systems (NeurIPS)}, 18, 1--8.

\bibitem[{Schoenberg(1935)}]{schoenberg1935remarks}
Schoenberg, I.~J. (1935), \enquote{{Remarks to Maurice Fr\'eechet's article
  ``Sur la d\'efinition axiomatique d'une classe d'espace distances
  vectoriellement applicable sur l'espace de Hilbert''},} \textit{Annals of
  Mathematics}, 724--732.

\bibitem[{Schoenberg(1937)}]{schoenberg1937certain}
--- (1937), \enquote{{On certain metric spaces arising from Euclidean spaces by
  a change of metric and their imbedding in Hilbert space},} \textit{Annals of
  Mathematics}, 38, 787--793.

\bibitem[{Schweinberger and Snijders(2003)}]{schweinberger2003settings}
Schweinberger, M.,  and Snijders, T.~A. (2003), \enquote{{Settings in social
  networks: A measurement model},} \textit{Sociological Methodology}, 33,
  307--341.

\bibitem[{Sewell and Chen(2015)}]{sewell2015latent}
Sewell, D.~K.,  and Chen, Y. (2015), \enquote{Latent space models for dynamic
  networks,} \textit{Journal of the American Statistical Association}, 110,
  1646--1657.

\bibitem[{Smith et~al.(2019)Smith, Asta, and Calder}]{smith2019geometry}
Smith, A.~L., Asta, D.~M.,  and Calder, C.~A. (2019), \enquote{The geometry of
  continuous latent space models for network data,} \textit{Statistical
  Science}, 34, 428--453.

\bibitem[{Smith(2020{\natexlab{a}})}]{smith2020information}
Smith, M.~R. (2020{\natexlab{a}}), \enquote{{Information theoretic generalized
  Robinson--Foulds metrics for comparing phylogenetic trees},}
  \textit{Bioinformatics}, 36, 5007--5013.

\bibitem[{Smith(2020{\natexlab{b}})}]{TreeDist}
--- (2020{\natexlab{b}}), \textit{{TreeDist: Distances between Phylogenetic
  Trees. R package version 2.6.1}}.

\bibitem[{Wang et~al.(2019)Wang, Zhang, and Dunson}]{wang2019common}
Wang, L., Zhang, Z.,  and Dunson, D. (2019), \enquote{Common and individual
  structure of brain networks,} \textit{The Annals of Applied Statistics}, 13,
  85--112.

\bibitem[{Yule(1925)}]{yule1925ii}
Yule, G.~U. (1925), \enquote{A mathematical theory of evolution, based on the
  conclusions of Dr. JC Willis, FR S,} \textit{Philosophical Transactions of
  the Royal Society of London. Series B}, 213, 21--87.

\end{thebibliography}
\endgroup

\end{document}